\documentclass[aps,prresearch,superscriptaddress,twocolumn,10pt,floatfix]{revtex4-2}

\usepackage{extarrows}
\usepackage{mathtools}
\usepackage{algpseudocode,algorithmicx,algorithm}
\usepackage{array}
\usepackage{booktabs}
\usepackage{etoolbox}
\usepackage{tcolorbox}
\usepackage[utf8]{inputenc}
\usepackage[T1]{fontenc}
\usepackage{graphicx}
\usepackage{amsthm}
\usepackage{amsmath}
\usepackage{mathdots}
\usepackage{verbatim}
\usepackage{pifont,url}
\usepackage{enumitem,amssymb}
\usepackage{multirow}
\usepackage{appendix}
\usepackage{epstopdf}
\usepackage{booktabs}
\usepackage{makecell}
\usepackage{diagbox}
\usepackage{setspace}
\usepackage{indentfirst}
\usepackage{qcircuit}

\usepackage{mathrsfs}
\usepackage{latexsym,bm}
\usepackage{amsmath,amsfonts,amsmath,amssymb,amsthm,bbm}
\usepackage{extarrows}

\usepackage{enumerate,cases,multirow}

\usepackage{hyperref}
\hypersetup{colorlinks=true,linkcolor=blue,anchorcolor=blue,citecolor=blue,urlcolor=blue}
\definecolor{mygray}{gray}{.9}

\usepackage{listings}
\usepackage{makeidx}        
\usepackage{booktabs}

\newcommand{\ket}[1]{| #1 \rangle} 
\newcommand{\bra}[1]{\langle #1 |} 

\def \d {\mathrm{d}}

\def \polylog {\mathrm{polylog}}

\usepackage{pifont}

\newtheorem{theoremmain}{Theorem}
\newtheorem{lemmamain}{Lemma}

\newtheorem{remarkmain}{Remark}
\newtheorem{theorem}{Theorem}[section]
\newtheorem{lemma}{Lemma}[section]

\theoremstyle{remark}
\newtheorem{remark}{\bf Remark}[section]

\allowdisplaybreaks[4]

\makeatletter
\let\oldaddcontentsline\addcontentsline

\renewcommand{\addcontentsline}[3]{
  \ifnum\pdfstrcmp{#2}{subsubsection}=0\relax
  \else
    \oldaddcontentsline{#1}{#2}{#3}
  \fi
}
\makeatother

\begin{document}
\title{Quantum Simulation of Non-Unitary Dynamics via Amplitude-Phase Separation}
\author{Qitong Hu}
\email{huqitong@sjtu.edu.cn}
\affiliation{School of Mathematical Sciences, Shanghai Jiao Tong University, Shanghai, 200240, China}
\author{Shi Jin}
\email{shijin-m@sjtu.edu.cn}
\affiliation{School of Mathematical Sciences, Shanghai Jiao Tong University, Shanghai, 200240, China}
\affiliation{Institute of Natural Sciences, Shanghai Jiao Tong University, Shanghai, 200240, China}
\affiliation{Ministry of Education (MOE) Funded Key Lab of Scientific and Engineering Computing, Shanghai Jiao Tong University, Shanghai, 200240, China}
\begin{abstract}
\par Linear non-unitary dynamics arise in open quantum systems, non-Hermitian models, and numerical evolution problems, yet current quantum algorithms do not cleanly separate coherent and dissipative effects at the design level. We introduce {\it Amplitude-Phase Separation} (APS), a decomposition framework with two complementary forms: phase-driven APS isolates the unitary component and maps the remainder to a Hermitian problem, whereas amplitude-driven APS extracts the Hermitian component and treats the remaining interaction separately. For time-independent dynamics, the two routes capture complementary advantages within one framework: phase-driven APS yields additive rather than multiplicative tolerance dependence, while amplitude-driven APS yields square-root dissipative scaling in multiscale regimes. APS also provides a unified interpretation of representative methods, including LCHS (Linear Combination of Hamiltonian Simulation) and NDME (Non-Diagonal Density Matrix Encoding), and clarifies where coherent and dissipative bottlenecks enter non-unitary simulation. The benchmarks confirm the predicted crossover between phase-driven and amplitude-driven advantages in advection-diffusion and Bloch-relaxation models.
\end{abstract}
\maketitle
\section{Introduction}
\label{section:1}
Quantum computers natively implement unitary dynamics governed by the Schr\"odinger equation. Many relevant processes, however, are effectively non-unitary, including open-system evolution, effective non-Hermitian models, and numerical evolution operators for PDEs and ODEs. Efficient quantum algorithms for these dynamics are therefore important in quantum information, open-system physics, and scientific computing~\cite{Lloyd1996UniversalQS,Childs2012Hamiltonian,Daley2022Nature,Babbush2023Oscillators,Zhang2025Eigen,Li2025Lattice}.

\par Existing methods fall mainly into two families. Hamiltonian-based approaches, including Schr\"odingerization~\cite{Jin2024Schrodingerization} and LCHS~\cite{An2023QuantumAF}, convert non-unitary tasks into unitary simulation and have markedly improved tolerance dependence~\cite{An2023Optimal,Low2025Optimal,Jin2025Optimal}. Lindbladian-based approaches, represented by NDME~\cite{Shang2025Lindbladians}, embed linear dynamics into open-system evolution and connect more directly to dissipative physics. Both routes are effective, but they optimize different resources and do not by themselves provide a unified operator-level organization for general non-unitary evolution.

\subsection{Problem and Gap}
\label{section:1-1}
\par We consider linear, non-autonomous (time-dependent), non-unitary dynamics that arise in open-system physics and scientific-computing problems such as partial differential equations (PDEs) and ordinary differential equations (ODEs):
\begin{equation}
    \begin{aligned}
    \label{equ:1}
        \frac{\d u(t)}{\d t}=-A(t)u(t)+b(t),\quad t\in[0,T]; \qquad u(0)=u_0,
    \end{aligned}
\end{equation}
where $A(t)\in\mathbb{C}^{n\times n}$ is a time-dependent matrix and $b(t)\in\mathbb{C}^{n\times 1}$ is a time-dependent vector. Here the system may be dissipative, and because $A(t)$ depends on time, the solution to Eq.~(\ref{equ:1}) at time $t$ is represented by a time-ordered exponential operator. We express this operator through the Cartesian decomposition~\cite{Soderlind2024Norm}, which writes any complex matrix $A(t)\in \mathbb{C}^{n \times n}$ as $A(t)=A_1(t)+iA_2(t)$, with Hermitian part $A_1(t)=\frac{A(t)+A^\dagger(t)}{2}\succeq 0$ and anti-Hermitian part $A_2(t)=\frac{A(t)-A^\dagger(t)}{2i}$.
For clarity, we refer to the operators associated with general $A(t)$, $A_1(t)$, and $A_2(t)$ as the {\it non-unitary}, {\it Hermitian}, and {\it unitary} operators, respectively.
In the special {\it time-independent} case $A(t) = A$, Low et al.~\cite{Low2025Optimal} conjectured that the optimal query complexity for block encoding $e^{-At}$ should scale as
\begin{equation}
    \begin{aligned}
        \label{equ:2}
        \mathcal{O}\Big(\underbrace{\sqrt{A_{1,\max}t\log\varepsilon^{-1}}}_{Fast-Forwarding}+\underbrace{A_{2,\max}t+\log\varepsilon^{-1}\vphantom{\sqrt{A_{1,\max}t\log\varepsilon^{-1}}}}_{Non-Unitary\ Optimality}\Big),
    \end{aligned}
\end{equation}
with $A_{1,\max}=\|A_1\|$ and $A_{2,\max}=\|A_2\|$. Here, the square-root $\sqrt{A_{1,\max}t}$ dependence is often referred to as non-unitary fast-forwarding~\cite{Low2017Forwarding,Shang2025Forwarding}, while the additive dependence on the error tolerance $\varepsilon$ in the $A_{2,\max}t$ term is regarded as non-unitary optimality~\cite{Low2017Optimal,Gilyen2018singular}, meaning that the tolerance contribution enters additively rather than multiplicatively in the leading query bound; this behavior arises from the no-fast-forwarding theorem~\cite{Berry2015Taylor,Berry2007Hamiltonians,Berry2015Hamiltonian}.

\par A natural question is whether one can combine the main strengths of these lines in a single construction. Current evidence suggests that this is nontrivial. Hamiltonian-based constructions inherit no-fast-forwarding constraints and may carry multiplicative tolerance overheads in broad settings (see Section~\ref{section:6-1}). Lindbladian-based constructions provide strong dissipative tools, but available fast-forwarding guarantees remain concentrated in time-independent purely dissipative regimes~\cite{Shang2025ExpLindbladian,Shang2025Forwarding}. The challenge is therefore not only to sharpen existing primitives, but also to find a decomposition that cleanly separates coherent and dissipative roles.

\par A key limitation is that the Cartesian decomposition is often used as an algebraic identity rather than as a guide for algorithm design. Writing $A(t)=A_1(t)+iA_2(t)$ does not by itself specify which simulation primitive should handle each part or which complexity barrier each part should inherit. Recent work~\cite{Jin2026Transmutation} revisits the special case of a {\it time-independent normal} matrix $A$ and derives the exact factorization $e^{-At}=e^{-A_1t}e^{-iA_2t}$. This observation suggests a broader design principle: the unitary factor can leverage Hamiltonian-optimal simulation, while the dissipative factor can leverage dissipative fast-forwarding tools. In that normal time-independent setting, this picture is consistent with the scaling target in Eq.~(\ref{equ:2}). This motivates the following question for {\it general time-dependent (non-autonomous)} systems:

\par {\it Can a general non-unitary evolution operator $\mathcal{T} e^{-\int_0^t A(s)\d s}$ be decomposed into a Hermitian part and a unitary part so as to achieve either fast-forwarding or non-unitary optimality as in Eq.~(\ref{equ:2})?}

\subsection{APS and Main Results}
\label{section:1-2}
\par Our answer is {\it Amplitude-Phase Separation} (APS), which turns the Cartesian decomposition into an algorithmic method for time-dependent non-unitary evolution. Phase-driven APS writes $\mathcal{T}e^{-\int_0^tA(s)\d s}$ as a unitary phase operator $\mathcal{T}e^{-i\int_0^tA_2(s)\d s}$ followed by a Hermitian amplitude operator, while amplitude-driven APS extracts the Hermitian factor $\mathcal{T}e^{-\int_0^tA_1(s)\d s}$ and treats the residual interaction separately. Compared with existing polar decomposition methods~\cite{Lloyd2020Polar,Quek2022Polar}, APS is designed for broader non-unitary and open-system settings.

\par APS applies generally as a decomposition, while the strongest complexity results proved here concern time-independent dynamics. In the considered oracle model, phase-driven APS gives additive rather than multiplicative tolerance overhead compatible with no-fast-forwarding limits from Hamiltonian simulation~\cite{Low2017Optimal,Gilyen2018singular,Berry2015Taylor} (Section~\ref{section:3}). Amplitude-driven APS gives square-root dissipative scaling in regimes where dissipation dominates ($A_{1,\max}\gg A_{2,\max}$), extending beyond purely Hermitian evolution (Section~\ref{section:4}). Table \ref{table:1} summarizes the corresponding assumptions and query complexities.

\par To connect these complexity statements with concrete model problems, the benchmarks in Section~\ref{section:5} confirm the predicted crossover between phase-driven and amplitude-driven advantages on a one-dimensional advection-diffusion problem and a complementary Bloch-type relaxation model.

\par APS also clarifies how Hamiltonian-based and dissipative-simulation approaches fit together. It recovers LCHS as phase-driven APS plus a generalized Fourier transform for $e^{-x}$ ($x\ge0$) applied to time-ordered Hermitian operators (see Section~\ref{section:6-1}), while NDME can be reinterpreted through phase-driven APS and the interaction picture for Lindbladians (see Section~\ref{section:6-2}). This makes explicit where coherent and dissipative resources enter non-unitary simulation and where further improvements would have to come from.

\subsection{Key Contributions}
\label{section:1-3}
\begin{itemize}[leftmargin=*]
\item \textbf{APS serves as a unified design principle for non-unitary simulation.} Rather than only producing two new complexity expressions, the paper places Hamiltonian and dissipative routes in one operator-level framework and shows when phase-driven or amplitude-driven organization is the natural choice.
\item \textbf{Phase-driven APS gives additive tolerance overhead in the time-independent setting.} In the stated oracle model, phase-driven APS plus shifted Dyson achieves $\mathcal{O}\!\left(\frac{\|u_0\|}{\|u(t)\|}\left(A_{\max}t+\frac{\log\varepsilon^{-1}}{\log\log\varepsilon^{-1}}\right)\right)$, so the tolerance contribution enters additively rather than multiplicatively in the leading query bound.
\item \textbf{Amplitude-driven APS identifies the dissipative regime with a square-root speedup.} In the stated oracle model, amplitude-driven APS yields $\mathcal{O}\!\left(\frac{\|u_0\|}{\|u(t)\|}\left(\sqrt{A_{1,\max}t\log\varepsilon^{-1}}+A_{2,\max}t\frac{\log\varepsilon^{-1}}{\log\log\varepsilon^{-1}}\right)\right)$, giving square-root scaling in a broader non-unitary setting than purely Hermitian evolution when dissipation dominates, and the benchmarks confirm the predicted crossover between phase-driven and amplitude-driven advantages.
\item \textbf{APS reinterprets representative methods within one structural framework.} APS recovers LCHS and NDME from the same decomposition viewpoint and clarifies how coherent and dissipative ingredients enter these constructions.
\end{itemize}

\begin{table*}[!tbp]
\centerline{
\resizebox{\textwidth}{!}{
\begin{tabular}{c|c|c|c|c}
\hline
\hline
    \makecell*[c]{Type}&
    \makecell*[c]{Method}&
    \makecell*[c]{Overall Queries}&
    \makecell*[c]{Applicable Operator}&
    \makecell*[c]{Queries to $u_0$}\\
\hline
    \multirow{4}{*}{\vspace{-3.75em}\makecell*[c]{Hamiltonian\\ Simulation}}&
    \makecell*[c]{$p$-th order Trotter~\cite{Berry2007Hamiltonians}}&
    \makecell*[c]{$\mathcal{O}\left(H_{\max}t(H_{\max}t\varepsilon^{-1})^{\frac{1}{p}}\right)$}&
    \makecell*[c]{Unitary}&
    \multirow{4}{*}{\vspace{-3.75em}\makecell*[c]{$\mathcal{O}\left(1\right)$}}\\
\cline{2-4}
    &\makecell*[c]{Truncated Dyson~\cite{Berry2015Taylor}}&
    \makecell*[c]{$\mathcal{O}\left(H_{\max}t\frac{\log\varepsilon^{-1}}{\log\log\varepsilon^{-1}}\right)$}&\makecell*[c]{Unitary}\\
\cline{2-4}
    &\makecell*[c]{Signal Processing~\cite{Low2017Optimal}}&
    \multirow{2}{*}{\vspace{-2em}\makecell*[c]{$\mathcal{O}\left(H_{\max}t +\frac{\log\varepsilon^{-1}}{\log\log\varepsilon^{-1}}\right)$}}&\multirow{2}{*}{\vspace{-2em}\makecell*[c]{Time-Independent\\ Unitary}}\\
\cline{2-2}
    &\makecell*[c]{Quantum Singular\\ Value Transformation~\cite{Gilyen2018singular}}&&\\
\hline
\hline
    \multirow{4}{*}{\vspace{-5.5em}\makecell*[c]{Lindbladian\\ Simulation}}&
    \makecell*[c]{Church-Turing\\ Theorem~\cite{Kliesch2011Church}}&
    \makecell*[c]{$\mathcal{O}\left(L_{\max}^2T^2\varepsilon^{-1}\right)$}&
    \makecell*[c]{Hermitian}&
    \multirow{4}{*}{\vspace{-5.5em}\makecell*[c]{$-$}}\\
\cline{2-4}
    &
    \makecell*[c]{Higher-Order Series\\ Expansion~\cite{Li2023ICALP}}&
    \makecell*[c]{$\mathcal{O}\left(L_{\max}t\frac{\log\varepsilon^{-1}}{\log\log\varepsilon^{-1}}\right)$}&\makecell*[c]{Time-Independent\\ Hermitian}\\
\cline{2-4}
    &\makecell*[c]{Shifted Technique~\cite{Shang2025ExpLindbladian}}&
    \makecell*[c]{$\mathcal{O}\left(L_{\max}t +\frac{\log\varepsilon^{-1}}{\log\log\varepsilon^{-1}}\right)$}&\makecell*[c]{Time-Independent\\ Hermitian}\\
\cline{2-4}
    &\makecell*[c]{Quantum Phase\\ Estimation~\cite{Shang2025Forwarding}}&
    \makecell*[c]{$\mathcal{O}\left(\sqrt{L_{\max}t\log\varepsilon^{-1}}\right)$}&\makecell*[c]{Time-Independent\\ Hermitian}\\
\hline
\hline
    \multirow{7}{*}{\vspace{-21em}\makecell*[c]{Non-Unitary\\ Dynamics}}&\makecell*[c]{LCHS~\cite{An2023QuantumAF}}&
    \multirow{2}{*}{\makecell*[c]{\vspace{-0.75em}$\mathcal{O}\left(\frac{\|u_0\|}{\|u(t)\|}A_{\max}t \varepsilon^{-1}\right)$}}&
    \multirow{2}{*}{\makecell*[c]{\vspace{-0.75em}Non-Unitary}}&
    \multirow{7}{*}{\vspace{-21em}\makecell*[c]{$\mathcal{O}\left(\frac{\|u_0\|}{\|u(t)\|}\right)$}}\\
\cline{2-2}
    &\makecell*[c]{Schr\"odingerization~\cite{Jin2024Schrodingerization}}&&\\
\cline{2-4}
    &\multirow{2}{*}{\vspace{-1.5em}\makecell*[c]{Improved LCHS~\cite{An2023Optimal}}}&
    \makecell*[c]{$\mathcal{O}\left(\frac{\|u_0\|}{\|u(t)\|}A_{\max}t (\log\varepsilon^{-1})^{\frac{1}{\beta}}\right), \beta\in (0,1)$}&\makecell*[c]{Time-Independent\\ Non-Unitary}\\
\cline{3-4}
    &&
    \makecell*[c]{$\mathcal{O}\left(\frac{\|u_0\|}{\|u(t)\|}\|H\|t (\log\varepsilon^{-1})^{1+\frac{1}{\beta}}\right), \beta\in (0,1)$}&\makecell*[c]{Time-Dependent\\ Non-Unitary}\\
\cline{2-4}
    &\makecell*[c]{Optimal LCHS~\cite{Low2025Optimal}}&
    \multirow{2}{*}{\makecell*[c]{$\mathcal{O}\left(\frac{\|u_0\|}{\|u(t)\|}A_{\max}t \log\varepsilon^{-1}\right)$}}&\multirow{2}{*}{\makecell*[c]{Time-Independent\\ Non-Unitary}}\\
\cline{2-2}
    &\multirow{2}{*}{\vspace{-1.5em}\makecell*[c]{Optimal Schr\"odingerization~\cite{Jin2025Optimal}}}&&\\
\cline{3-4}
    &&
    \makecell*[c]{$\mathcal{O}\left(\frac{\|u_0\|}{\|u(t)\|}A_{\max}t (\log\varepsilon^{-1})^2\right)$}&\makecell*[c]{Time-Dependent\\ Non-Unitary}\\
\cline{2-4}
    &\multirow{2}{*}{\vspace{-1.5em}\makecell*[c]{Lindbladians~\cite{Shang2025Lindbladians}}}&
    \makecell*[c]{$\mathcal{O}\left(\frac{\|u_0\|}{\|u(t)\|}A_{\max}t\frac{\log\varepsilon^{-1}}{\log\log\varepsilon^{-1}}\right)$}&\makecell*[c]{Time-Independent\\ Non-Unitary}\\
\cline{3-4}
    &&
    \makecell*[c]{$\mathcal{O}\left(\frac{\|u_0\|}{\|u(t)\|}A_{\max}t\frac{(\log\varepsilon^{-1})^2}{(\log\log\varepsilon^{-1})^2}\right)$}&\makecell*[c]{Time-Dependent\\ Non-Unitary}\\
\cline{2-4}
    &\multirow{2}{*}{\vspace{-1.5em}\makecell*[c]{Phase-Driven APS\\ (This work)}}&
    \makecell*[c]{$\mathcal{O}\left(\frac{\|u_0\|}{\|u(t)\|}\left(A_{\max}t+\frac{\log\varepsilon^{-1}}{\log\log\varepsilon^{-1}}\right)\right)$}&\makecell*[c]{Time-Independent\\ Non-Unitary}\\
\cline{3-4}
    &&
    \makecell*[c]{$\mathcal{O}\left(\frac{\|u_0\|}{\|u(t)\|}A_{\max}t\frac{\log\varepsilon^{-1}}{\log\log\varepsilon^{-1}}\right)$}&\makecell*[c]{Time-Dependent\\ Non-Unitary}\\
\cline{2-4}
    &\makecell*[c]{Amplitude-Driven APS\\ (This work)}&
    \makecell*[c]{$\mathcal{O}\left(\frac{\|u_0\|}{\|u(t)\|}\left(\sqrt{A_{1,\max}t\log\varepsilon^{-1}}+A_{2,\max}t\frac{\log\varepsilon^{-1}}{\log\log\varepsilon^{-1}}\right)\right)$}&\makecell*[c]{Time-Independent\\ Non-Unitary}\\
\hline
\hline
\end{tabular}}}
\caption{\label{table:1}\textbf{Summary of Quantum Algorithms.} All complexity expressions in Table \ref{table:1} are stated under each method’s corresponding oracle-access model and assumptions; comparisons are made only under those stated settings.}
\vspace{-1.5em}
\end{table*}

\par The rest of this paper is organized as follows. Section~\ref{section:2} introduces the two APS factorizations. Section~\ref{section:3} establishes non-unitary optimality via phase-driven APS and the shifted Dyson series. Section~\ref{section:4} establishes fast-forwarding via amplitude-driven APS. Section~\ref{section:5} presents advection-diffusion and Bloch-type relaxation benchmarks. Section~\ref{section:6} discusses the relation to LCHS and NDME, together with current limitations and open directions.

\section{Amplitude-Phase Separation}
\label{section:2}
\par APS applies to general {\it time-dependent} (non-autonomous) non-unitary dynamics in Eq.~(\ref{equ:1}). It is closely related to exponential integrators~\cite{Hochbruck2010Exp} and the interaction picture~\cite{Low2019Interaction, Bosse2025Hamiltonian}, both of which preserve selected structures of the evolution. APS does the same at the operator level: phase-driven APS keeps the unitary part explicit, whereas amplitude-driven APS keeps the Hermitian part explicit.

\subsection{Phase-Driven APS}
\label{section:2-1}
\par We define the exponential integrator $\mathcal{U}_p(t)=\mathcal{T}e^{-i\int_0^t A_2(s)\d s}$ and set $u_p(t)=\mathcal{U}_p^\dagger(t) u(t)$, such that $u_p(0) = u_0$. The resulting dynamical system governing $u_p(t)$ is then given by
\begin{equation}
    \begin{aligned}
        \label{equ:3}
        \frac{\d u_p(t)}{\d t}=-A_p(t)u_p(t),
    \end{aligned}
\end{equation}
where $A_p(t)=\mathcal{U}_p^\dagger(t)A_1(t)\mathcal{U}_p(t)$ is a Hermitian matrix. Using the relation between $u_p(t)$ and $u(t)$, we obtain an equivalent representation of the operator $\mathcal{T}e^{-\int_0^t A(s)\d s}$. This representation is referred to as phase-driven APS:
\begin{equation}
    \begin{aligned}
        \label{equ:4}
        &\textit{Phase-Driven APS:}\\
        &\qquad\mathcal{T}e^{-\int_0^tA(s)\d s}=\mathcal{T}e^{-i\int_0^tA_2(s)\d s}\cdot\mathcal{T}e^{-\int_0^tA_p(s)\d s}.
    \end{aligned}
\end{equation}
This decomposes any non-unitary operator into a phase-amplitude form. The detailed proof is provided in Theorem~\ref{theorem:A1}.

\begin{figure*}[htbp]
    \centering
    \includegraphics[width=0.7\linewidth]{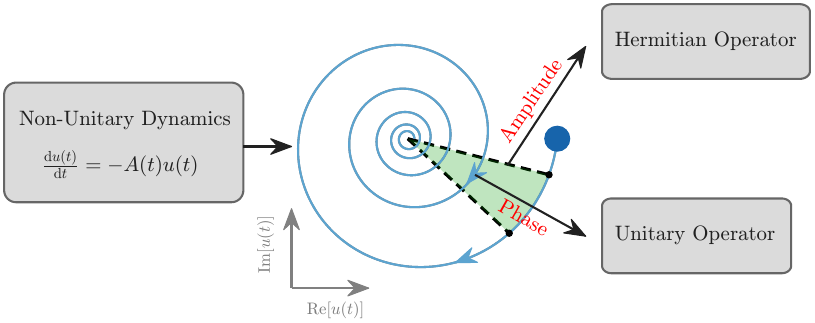}
    \caption{
    \label{figure:aps-overview}\textbf{Amplitude-Phase Separation (APS) for non-unitary dynamics.}}
\end{figure*}

\begin{theoremmain}
\label{theorem:A1}
\par Consider the following linear ODE with the initial condition $u(0)=u_0$:
\begin{equation}
    \begin{aligned}
        \label{equ:interaction:1}
        \frac{\d u(t)}{\d t}=-A(t)u(t)=-(A_1(t)+iA_2(t))u(t),\quad t\in[0,T].
    \end{aligned}
\end{equation}
Then, the following equivalent expression holds for the operator $\mathcal{T} e^{-\int_0^t A(s)\d s}$:
\begin{equation}
    \begin{aligned}
        \label{equ:interaction:2}
        \mathcal{T}e^{-\int_0^tA(s)\d s}=\mathcal{T}e^{-i\int_0^tA_2(s)\d s}\cdot\mathcal{T}e^{-\int_0^tA_p(s)\d s},
    \end{aligned}
\end{equation}
where $A_p(t)=\mathcal{U}_p^\dagger(t) A_1(t)\mathcal{U}_p(t)$ and $\mathcal{U}_p(t)=\mathcal{T}e^{-i\int_0^t A_2(s)\d s}$. It is evident that the eigenvalues of $A_p(t)$ are identical to those of $A_1(t)$. This implies that $A_p(t)$ is a positive semidefinite matrix.
\begin{proof}
\par First, we define the exponential integrator~\cite{Hochbruck2010Exp} $\mathcal{U}_p(t)=\mathcal{T}e^{-i\int_0^t A_2(s)\d s}$~\cite{Low2019Interaction,Bosse2025Hamiltonian}. Clearly, this operator is the solution to the following ODE:
\begin{equation*}
    \begin{aligned}
        \frac{\d \mathcal{U}_p(t)}{\d t}=-iA_2(t)\mathcal{U}_p(t),\quad \mathcal{U}_p(0)=I.
    \end{aligned}
\end{equation*}
Next, we let $u_p(t)=\mathcal{U}_p^\dagger(t) u(t)$ with $u_p(0) = u_0$. We can derive the ODE satisfied by $u_p(t)$ as follows:
\begin{equation}
    \begin{aligned}
        \label{equ:interaction:3}
        \frac{\d u_p(t)}{\d t}=&\frac{\d \mathcal{U}_p^\dagger(t)}{\d t}u(t)+\mathcal{U}_p^\dagger(t)\frac{\d u(t)}{\d t}\\
        =&\left[-\mathcal{U}_p^\dagger(t)\frac{d\mathcal{U}_p(t)}{\d t}\mathcal{U}_p^\dagger(t)-\mathcal{U}_p^\dagger(t)(A_1(t)+iA_2(t))\right]u(t)\\
        =&-\mathcal{U}_p^\dagger(t)A_1(t)\mathcal{U}_p(t)u_p(t)\\
        :=&-A_p(t)u_p(t),
    \end{aligned}
\end{equation}
where we define $A_p(t)=\mathcal{U}_p^\dagger(t)A_1(t)\mathcal{U}_p(t)$. Hence, we can transform the non-unitary dynamics in Eq.~(\ref{equ:interaction:1}) into purely dissipative, time-dependent dynamics in Eq.~(\ref{equ:interaction:3}). The state $u(t)$ is then expressed as:
\begin{equation*}
    \begin{aligned}
        u(t)&=\mathcal{U}_p(t)u_p(t)=\mathcal{T}e^{-i\int_0^tA_2(s)\d s}\cdot\mathcal{T}e^{-\int_0^tA_p(s)\d s}\cdot u_0.
    \end{aligned}
\end{equation*}
From this, we obtain the equivalent representation of the operator $\mathcal{T}e^{-\int_0^t A(s)\d s}$ as:
\begin{equation*}
    \begin{aligned}
        \mathcal{T}e^{-\int_0^tA(s)\d s}=\mathcal{T}e^{-i\int_0^tA_2(s)\d s}\cdot\mathcal{T}e^{-\int_0^tA_p(s)\d s}.
    \end{aligned}
\end{equation*}
Therefore, we decompose the simulation of a non-unitary operator into the separate simulations of a unitary operator $\mathcal{T}e^{-i\int_0^tA_2(s)\d s}$ and a Hermitian operator $\mathcal{T}e^{-\int_0^tA_p(s)\d s}$.
\end{proof}
\end{theoremmain}
\noindent Numerous optimal results already exist for simulating unitary operators, particularly in the time-independent case $A(t) = A$. Thus, we can focus mainly on approximating the Hermitian operator.
\subsubsection{Relation with Product Formula}
\par APS is directly connected to product-formula techniques~\cite{Bosse2025Hamiltonian}, a standard tool in quantum computing. We use the time-independent case of phase-driven APS for illustration. By treating $A_1$ and $iA_2$ as the two parts in a product formula, the following equality holds:
\begin{widetext}
\begin{equation}
    \begin{aligned}
        \label{equ:interaction:7}
        e^{-At}
        =\lim\limits_{n\to\infty}\left(e^{-i\frac{A_2t}{n}}e^{-\frac{A_1t}{n}}\right)^n
        &=e^{-iA_2t}\lim\limits_{n\to\infty}\prod\limits_{k=1}^ne^{i(1-\frac{k}{n})A_2t}e^{-\frac{A_1t}{n}}e^{-i(1-\frac{k}{n})A_2t}\\
        &=e^{-iA_2t}\lim\limits_{n\to\infty}\prod\limits_{k=1}^ne^{-e^{i(1-\frac{k}{n})A_2t}\cdot\frac{A_1t}{n}\cdot e^{-i(1-\frac{k}{n})A_2t}}\\
        &=e^{-iA_2t}\mathcal{T}e^{-\int_0^te^{iA_2s}\cdot A_1\cdot e^{-iA_2s}\d s}.
    \end{aligned}
\end{equation}
\end{widetext}
This corresponds to the phase-driven APS introduced in Theorem~\ref{theorem:A1}.
\subsection{Amplitude-Driven APS}
\label{section:2-2}
\par Similarly, amplitude-driven APS takes an amplitude-phase form. It is obtained by introducing the exponential integrator $\mathcal{U}_a(t)=\mathcal{T}e^{-\int_0^tA_1(s)\d s}$ and defining $A_a(t)=\mathcal{U}_a^{-1}(t) A_2(t)\mathcal{U}_a(t)$. It is given explicitly by:
\begin{equation}
    \begin{aligned}
        \label{equ:5}
        &\textit{Amplitude-Driven APS:}\\
        &\qquad\mathcal{T}e^{-\int_0^tA(s)\d s}=\mathcal{T}e^{-\int_0^tA_1(s)\d s}\cdot\mathcal{T}e^{-i\int_0^tA_a(s)\d s},
    \end{aligned}
\end{equation}
and its proof is provided in Theorem~\ref{theorem:A2}.
\begin{theoremmain}
\label{theorem:A2}
\par Consider the following linear ODE with the initial condition $u(0)=u_0$:
\begin{equation}
    \begin{aligned}
        \label{equ:interaction:4}
        \frac{\d u(t)}{\d t}=-A(t)u(t)=-(A_1(t)+iA_2(t))u(t),\quad t\in[0,T].
    \end{aligned}
\end{equation}
Then, the following equivalent expression holds for the operator $\mathcal{T}e^{-\int_0^tA(s)\d s}$:
\begin{equation}
    \begin{aligned}
        \label{equ:interaction:5}
        \mathcal{T}e^{-\int_0^tA(s)\d s}=\mathcal{T}e^{-\int_0^tA_1(s)\d s}\cdot\mathcal{T}e^{-i\int_0^tA_a(s)\d s},
    \end{aligned}
\end{equation}
where $A_a(t)=\mathcal{U}_a^{-1}(t) A_2(t)\mathcal{U}_a(t)$ and $\mathcal{U}_a(t)=\mathcal{T}e^{-\int_0^tA_1(s)\d s}$.
\begin{proof}
\par We set $\mathcal{U}_a(t)=\mathcal{T}e^{-\int_0^tA_1(s)\d s}$ and define $u_a(t)=\mathcal{U}_a^{-1}(t)u(t)$, such that $u_a(0)=u_0$. Then we have:
\begin{equation}
    \begin{aligned}
        \label{equ:interaction:6}
        \frac{\d u_a(t)}{\d t}=&\frac{\d \mathcal{U}_a^{-1}(t)}{\d t}u(t)+\mathcal{U}_a^{-1}(t)\frac{\d u(t)}{\d t}\\
        =&-i\mathcal{U}_a^{-1}(t)A_2(t)\mathcal{U}_a(t)u_a(t)\\
        :=&-iA_a(t)u_a(t),
    \end{aligned}
\end{equation}
where $A_a(t)=\mathcal{U}_a^{-1}(t) A_2(t)\mathcal{U}_a(t)$. Hence, we obtain the expression for $u(t)$ as:
\begin{equation*}
    \begin{aligned}
        u(t)&=\mathcal{U}_a(t)u_a(t)=\mathcal{T}e^{-\int_0^tA_1(s)\d s}\cdot\mathcal{T}e^{-i\int_0^tA_a(s)\d s}\cdot u_0.
    \end{aligned}
\end{equation*}
From this, the equivalent representation of the operator $\mathcal{T}e^{-\int_0^t A(s)\d s}$ is derived:
\begin{equation*}
    \begin{aligned}
        \mathcal{T}e^{-\int_0^tA(s)\d s}=\mathcal{T}e^{-\int_0^tA_1(s)\d s}\cdot\mathcal{T}e^{-i\int_0^tA_a(s)\d s}.
    \end{aligned}
\end{equation*}
\end{proof}
\end{theoremmain}
\noindent In the time-independent case, i.e., $A(t)=A$, we can decompose the operator $e^{-At}$ into a time-dependent operator involving $A_a(t)$ and a time-independent Hermitian operator involving $A_1$ through this factorization. The latter admits a fast-forwarding quantum implementation, achieving a square-root query complexity. Using this decomposition, the same bound can be achieved with respect to both matrix norm and evolution time.
\par Under the definition $A_a(t)=\mathcal{U}_a^{-1}(t)A_2(t)\mathcal{U}_a(t)$, $A_a(t)$ is generally non-Hermitian even when $A_2(t)$ is Hermitian. Therefore, the corresponding interaction factor $\mathcal{T}e^{-i\int_0^tA_a(s)\d s}$ should be treated as a general (not necessarily unitary) operator.
\medskip
\par Phase-driven APS separates the unitary and Hermitian parts of the system, while amplitude-driven APS extracts the Hermitian part from the dynamics, as illustrated in Fig.~\ref{figure:aps-overview}. This capability allows the two formulations to combine advantages and optimality guarantees from distinct approaches, thereby extending their applicability to non-unitary and time-dependent dynamics.
\par Notably, both APS formulations rely on time-dependent representations even when solving time-independent problems, which is counterintuitive. Since existing results for time-dependent unitary, Hermitian, and general non-unitary operators have not yet been unified optimally, we need different APS methods for different problem settings, as discussed below.

\section{Non-Unitary Optimality via Phase-Driven APS}
\label{section:3}
For Hamiltonian simulation, the optimal query complexity scales at least linearly with the evolution time $t$~\cite{Berry2015Taylor} and achieves an additive $\mathcal{O}(\log \varepsilon^{-1})$ dependence on the error tolerance $\varepsilon$~\cite{Low2017Optimal}. Specifically, techniques such as Quantum Signal Processing (QSP)~\cite{Low2017Optimal} and Quantum Singular Value Transformation (QSVT)~\cite{Gilyen2018singular} reach this lower bound for time-independent Hamiltonians. However, a significant gap remains when extending this optimality to {\it time-independent} non-unitary dynamics using existing methods, as outlined in Section~\ref{section:6}.

\subsection{Shifted Dyson Series}
\label{section:3-1}
\par We can bridge this gap using the phase-driven APS. By expressing $e^{-At}=e^{-iA_2t}\mathcal{T}e^{-\int_0^t A_p(s)\d s}$ with $A_p(s)=e^{iA_2s}A_1e^{-iA_2s}$, we effectively reformulate the simulation problem into implementing $\mathcal{T}e^{-\int_0^t A_p(s) \d s}$.
\par To handle this Hermitian operator, we introduce a {\it shift technique} into its Dyson series and express it as a Riemann sum (see Appendix~\ref{section:8-1}). This shift is analogous to expanding the function $e^{-x}$ for $x\in[0,1]$ around $1$ instead of $0$. It redistributes the eigenvalues more symmetrically around zero, significantly improving error control when truncating the series. This yields a truncation with the number of terms $M$ given by:
\begin{equation}
    \begin{aligned}
        \label{equ:6}
        M=\mathcal{O}\left(A_{\max}t+\frac{\log\varepsilon^{-1}}{\log\log\varepsilon^{-1}}\right),
    \end{aligned}
\end{equation}
where $t$ is the evolution time. The detailed proof is provided in Lemma~\ref{lemma:B1}.
\begin{lemmamain}
\label{lemma:B1}
\par Let $L(s)$ be a Hermitian matrix satisfying $L(s)\succeq 0$ for all $s\in[0, t]$. Define $L_{\max}=\sup\limits_{s\in[0, t]}\|L(s)\|$ and $\tau=tL_{\max}$. For any error tolerance $\varepsilon>0$, if the truncation order $M$ is chosen such that 
\begin{equation}
    \begin{aligned}
        \label{equ:dyson:2}
        M=\mathcal{O}\left(\tau+\frac{\log\varepsilon^{-1}}{\log\log\varepsilon^{-1}}\right),
    \end{aligned}
\end{equation}
with the implied multiplicative constant at most $e^2$, then the time-ordered Hermitian operator $\mathcal{T}e^{-\int_0^t L(s)\d s}$ can be approximated by a $M$-term series with error at most $\varepsilon$:
\begin{equation*}
    \begin{aligned}
        \left\|\mathcal{T}e^{-\int_0^tL(s)\d s}-\sum\limits_{k=0}^M(-1)^kP_k\right\|\le\varepsilon,
    \end{aligned}
\end{equation*}
where $P_k$ is the $k-$th term in the Dyson series, given by:
\begin{equation}
    \begin{aligned}
        \label{equ:dyson:3}
        P_k=\frac{e^{-\tau}}{k!}\int_0^t\cdots\int_0^t\mathcal{T}\left[\tilde{L}(t_1)\cdots\tilde{L}(t_k)\right]\d^k t,
    \end{aligned}
\end{equation}
in which $\tilde{L}(t)=L(t)-L_{\max}I$.
\begin{proof}
\par Note that here we do not follow the standard Dyson series directly; instead, we first perform the following transformation before expanding:
\begin{equation*}
    \begin{aligned}
        &\mathcal{T}e^{-\int_0^tL(s)\d s}=e^{-\tau}\mathcal{T}e^{-\int_0^t\tilde{L}(s)\d s}\\
        =&e^{-\tau}\sum\limits_{k=0}^\infty \frac{(-1)^k}{k!}\int_0^t\cdots\int_0^t\mathcal{T}\left[\tilde{L}(t_1)\cdots\tilde{L}(t_k)\right]\d^k t.
    \end{aligned}
\end{equation*}
Then, based on the definition of $P_k$ proposed in Eq.~(\ref{equ:dyson:3}) and the properties of the time-ordering operator, we can derive the following upper bound estimate for $\|P_k\|$:
\begin{equation*}
    \begin{aligned}
        \|P_k\|&\le \frac{e^{-\tau}}{k!}\int_0^t\cdots\int_0^t\prod_{j=1}^k\left\|\tilde{L}(t_k)\right\|\d^k t\le \frac{e^{-\tau}\tau^k}{k!},
    \end{aligned}
\end{equation*}
where we have used the inequality $-L_{\max} I\preceq L(t)-L_{\max} I\preceq O$, which implies $\|L(t)-L_{\max} I\|\le L_{\max}$. Note that the upper bound of $\|P_k\|$ is essentially controlled by a Poisson distribution with parameter $\tau$.
\par Therefore, to control the tail distribution with error tolerance $\varepsilon$, under the condition $M>e^2\tau$, we obtain the following general Chernoff-type bound:
\begin{equation}
    \begin{aligned}
        \label{equ:dyson:4}
        &\left\|\mathcal{T}e^{-\int_0^tL(s)\d s}\right.-\left.\sum\limits_{k=0}^M(-1)^kP_k\right\|\le\sum\limits_{k=M+1}^\infty\left\|P_k\right\|\\
        \le& \sum\limits_{k=M+1}^\infty\frac{e^{-\tau}\tau^k}{k!}
        \le \sum\limits_{k=M+1}^\infty e^{-\tau}\left(\frac{e\tau}{k}\right)^k\\
        \le& \sum\limits_{k=M+1}^\infty e^{-(\tau+k)}\le e^{-(\tau+M)},
    \end{aligned}
\end{equation}
which implies that the error can be effectively controlled if we choose $M$ such that
\begin{equation}
    \begin{aligned}
        \label{equ:dyson:5}
        M>\max\left\{e^2\tau,\log\varepsilon^{-1}\right\}=\mathcal{O}\left(\tau+\log\varepsilon^{-1}\right),\quad \forall \tau,
    \end{aligned}
\end{equation}
where the implied multiplicative constant is at most $e^2$. This bound is well-suited for the case when $\tau\ge\frac{\log \varepsilon^{-1}}{e+1}$. However, for the complementary case, a more precise estimate is required.
\par For $\tau\le\frac{\log\varepsilon^{-1}}{e+1}$, following the procedure in Eq.~(\ref{equ:dyson:4}) and under the condition that $M>2\tau$, we obtain the following upper bound estimate:
\begin{equation}
    \begin{aligned}
        \label{equ:dyson:6}
        &\left\|\mathcal{T}e^{-\int_0^tL(s)\d s}\right.-\left.\sum\limits_{k=0}^M(-1)^kP_k\right\|\le\sum\limits_{k=M+1}^\infty\left\|P_k\right\|\\
        \le& \sum\limits_{k=M}^\infty\frac{e^{-\tau}\tau^k}{k!}\le \frac{e^{-\tau}\tau^k}{M!}\sum\limits_{k=M}^\infty\frac{1}{2^{k-M}}\le e^{-\tau}\left(\frac{e\tau}{M}\right)^M,
    \end{aligned}
\end{equation}
where this equation can be viewed as a transformation involving the Lambert function $\mathcal{W}(x)$ (whose definition is $\mathcal{W}(x)e^{\mathcal{W}(x)}=x$ and is monotonically increasing when $x\ge -\frac{1}{e}$~\cite{Abramowitz1965Handbook,Olver2010NIST}), and we can derive a lower bound for $M$ by imposing the condition $e^{-\tau}\left(\frac{e\tau}{M}\right)^M \le \varepsilon$~\cite{Low2019Interaction}, which is equivalent to $\left(\frac{e\tau}{M}\right)^{\frac{M}{e\tau}} \le (e^\tau\varepsilon)^{\frac{1}{e\tau}}$:
\begin{equation*}
    \begin{aligned}
        M&\ge\max\left\{2\tau,\tau\frac{-1+\frac{\log\varepsilon^{-1}}{\tau}}{\mathcal{W}(\frac{1}{e}(-1+\frac{\log\varepsilon^{-1}}{\tau}))}\right\}\\
        :&=\max\left\{2\tau,
        \tau \frac{ex}{\mathcal{W}(x)}\right\},
    \end{aligned}
\end{equation*}
where we let $x := \frac{1}{e}(-1+\frac{\log\varepsilon^{-1}}{\tau})$ and the basic condition $x>-\frac{1}{e}$ for the Lambert function $\mathcal{W}(x)$ is clearly satisfied~\cite{Abramowitz1965Handbook,Olver2010NIST}. Furthermore, we will use Eqs. (\ref{equ:dyson:5}) to derive a simpler and tighter lower bound for $M$ under the condition $\tau \le \frac{\log \varepsilon^{-1}}{e+1}$ with $x\ge 1$. Using the inequality $\mathcal{W}(x) \ge \frac{\log x + 1}{2}$ for $x \ge 1$~\cite{Low2017Optimal}, we can establish a lower bound for $M$ via the third term in Eqs. (\ref{equ:dyson:5}):
\begin{equation*}
    \begin{aligned}
        M\ge \left\{2\tau,2\tau\frac{-1+\frac{\log\varepsilon^{-1}}{\tau}}{\log(-1+\frac{\log\varepsilon^{-1}}{\tau})}\right\},
    \end{aligned}
\end{equation*}
where by substituting the condition $\tau \le \frac{\log \varepsilon^{-1}}{e+1}$, we can find a relatively tight upper bound for the second term. Using the inequality $\frac{a+b}{c+b}\ge\frac{a}{c}$ for $a,b,c >0$, and leveraging the fact that $\tau-\log\tau+\log\frac{e}{e+1}>0$ for $\tau>0$ (with a minimum value of $1+\log \frac{e}{e+1}>0$), we obtain
\begin{equation*}
    \begin{aligned}
        &2\frac{-\tau+\log\varepsilon^{-1}}{\log(-1+\frac{\log\varepsilon^{-1}}{\tau})}\le\frac{2\log\varepsilon^{-1}}{\tau-\log\tau+\log(\log\varepsilon^{-1}-\tau)}\\
        \le& \frac{2\log\varepsilon^{-1}}{\tau-\log\tau+\log\frac{e}{e+1}+\log(\log\varepsilon^{-1})}\le \frac{2\log\varepsilon^{-1}}{\log\log\varepsilon^{-1}},
    \end{aligned}
\end{equation*}
which implies that, for $\tau<\frac{\log\varepsilon^{-1}}{e+1}$, the upper bound for $M$ can be set as 
\begin{equation}
    \begin{aligned}
        \label{equ:dyson:7}
        M\ge \max\left\{2\tau,\frac{2\log\varepsilon^{-1}}{\log\log\varepsilon^{-1}}\right\}=\mathcal{O}\left(\tau+\frac{\log\varepsilon^{-1}}{\log\log\varepsilon^{-1}}\right).
    \end{aligned}
\end{equation}
\par Combining the results from Eqs. (\ref{equ:dyson:5}) and (\ref{equ:dyson:7}) with the constraint $M\ge2\tau$, we obtain a lower bound that $M$ must satisfy to meet the error requirement:
\begin{equation*}
    \begin{aligned}
        M&=\left\{\begin{array}{l}
        \mathcal{O}\left(\tau+\log\varepsilon^{-1}\right),\quad \tau\ge\frac{\log\varepsilon^{-1}}{e+1},\\
        \mathcal{O}\left(\tau+\frac{\log\varepsilon^{-1}}{\log\log\varepsilon^{-1}}\right),\tau<\frac{\log\varepsilon^{-1}}{e+1}.\end{array}\right\}\\
        &=\mathcal{O}\left(\tau+\frac{\log\varepsilon^{-1}}{\log\log\varepsilon^{-1}}\right),
    \end{aligned}
\end{equation*}
where the implied multiplicative constant is at most $e^2$. This completes the proof.
\end{proof}
\end{lemmamain}
\par In Lemma~\ref{lemma:B1}, we provided an integral representation for approximating Hermitian operators. Its discrete Riemann-sum form is further derived in Appendix~\ref{section:8-1}. Denoting the error in Lemma~\ref{lemma:B1} by $\varepsilon_1$ and the error in Lemma~\ref{lemma:B2} by $\varepsilon_2$, with $\varepsilon_1 + \varepsilon_2 = \varepsilon$, we can state the following Theorem~directly.

\begin{theoremmain}
\label{theorem:B1}
\par Let $L(s)$ be a Hermitian matrix satisfying $L(s)\succeq 0$ for all $s\in[0,t]$. Define $L_{\max}=\sup\limits_{s\in[0, t]}\|L(s)\|$, $\tau=tL_{\max}$, and $\hat{L}_{\max}=\sup\limits_{s\in[0, t]}\left\|\left.\frac{\d L(t)}{\d t}\right|_{t=s}\right\|$. For any error tolerance $\varepsilon > 0$, partition the time interval $[0, t]$ into $N_m$ points, and define the time step $h = \frac{t}{N_m}$. If the truncation order $M$ and the number of discrete intervals $N_m$ are chosen such that
\begin{equation}
    \begin{aligned}
    \label{equ:dyson:1}
    M=\mathcal{O}\left(\tau+\frac{\log\varepsilon^{-1}}{\log\log\varepsilon^{-1}}\right),N_m\ge \frac{M(\hat{L}_{\max}t+L_{\max}(\tau+1))}{\varepsilon},
    \end{aligned}
\end{equation}
then the time-ordered Hermitian operator $\mathcal{T}e^{-\int_0^t L(s)\d s}$ can be approximated by the following series with error at most $\varepsilon$:
\begin{equation*}
    \begin{aligned}
        \left\|\mathcal{T}e^{-\int_0^t L(s)\d s}-\sum\limits_{k=0}^M(-1)^kQ_k\right\|\le\varepsilon,
    \end{aligned}
\end{equation*}
where $Q_k$, which does not contain the time-ordering operator $\mathcal{T}$, is defined as
\begin{equation*}
    \begin{aligned}
        Q_k=e^{-\tau}h^k\sum\limits_{0\le m_1<\cdots<m_k< N_m}\tilde{L}(m_kh)\cdots\tilde{L}(m_1h),
    \end{aligned}
\end{equation*}
in which $\tilde{L}(t)=L(t)-L_{\max}I$ and $m_1,\cdots,m_k$ are indices in $0,\cdots,N_m$.
\end{theoremmain}

\begin{remarkmain}
\label{remark:B3}
\par Consider the LCU $I_M=\sum\limits_{k=1}^M \alpha_k U^k$, where $U$ is a Hamiltonian satisfying $\|U\| \le 1$. If $\sum\limits_{k=1}^M|\alpha_k|\le C$ with $C$ being a constant, then ensuring $I_M$ stays within an error tolerance $\varepsilon$ requires the approximation error for $U$ to satisfy $\varepsilon^\prime = \frac{\varepsilon}{CM}$. This holds because
\begin{equation}
    \begin{aligned}
        \left\|\sum\limits_{k=1}^M \alpha_k (U + \varepsilon^\prime)^k - \sum\limits_{k=1}^M \alpha_k U^k \right\| \le \sum\limits_{k=1}^M k |\alpha_k| \varepsilon^\prime \le C M \varepsilon^\prime.
    \end{aligned}
\end{equation}
\end{remarkmain}
\noindent First, we can verify that the condition in Theorem~\ref{theorem:B1} is satisfied because  
\begin{equation*}
    \begin{aligned}
        c=\sum\limits_{k=0}^Me^{-\tau}h^k\sum\limits_{0\le m_1<\cdots<m_k< N_m}&L_{\max}^k\le \sum\limits_{k=0}^Me^{-\tau}\frac{\tau^k}{k!} \le 1.
    \end{aligned}
\end{equation*}
This condition is achievable because of the constant term $e^{-\tau}$, introduced by the shift technique, which suppresses the truncation error.
\par This worst-case bound already suffices to support Theorem~\ref{theorem:main:1}. Appendix~\ref{section:8-2} explains why a dissipative contribution is necessary for the shifted construction, while Appendix~\ref{section:8-3} proves a sharper regime-dependent behavior beyond this worst-case statement.

\par In particular, when $M \ge 2\tau$ and $\tau$ is larger, corresponding to a stronger dissipative regime, the shifted Dyson series can achieve $M \sim (\log \varepsilon^{-1})^{\frac{1}{2}}$. That is, the truncation count exhibits square-root dependence on the logarithmic error parameter; the corresponding numerical verification is also given in Appendix~\ref{section:8-3}.

\subsection{Application of Phase-Driven APS}
\label{section:3-2}
\par By implementing the Hermitian operator $\mathcal{T}e^{-\int_0^t A_p(s) \d s}$ via the shifted Dyson series, we obtain an approximation of $e^{-At}$ within an error tolerance $\varepsilon$:
\begin{equation}
    \begin{aligned}
        \label{equ:7}
        e^{-At}&\approx e^{-\tau}\sum\limits_{k=0}^M(-1)^kh^k\\
        &\quad\times\sum\limits_{0\le m_1<\cdots<m_k< N_m}e^{-iA_2t}\overleftarrow{\prod_{j=1}^k}e^{iA_2m_jh}\tilde{A}_1e^{-iA_2m_jh},
    \end{aligned}
\end{equation}
where $\tilde{A}_1=A_1-A_{\max} I$ is the shifted Hermitian matrix with $\|\tilde{A}_1\|\le A_{\max}$, and $N_m=\mathcal{O}(MA_{\max}^2t\varepsilon^{-1})$ denotes the number of discrete subintervals of $[0,t]$ with step size $h=\frac{t}{N_m}$. The indices $m_1,\cdots,m_k$ range from $0$ to $N_m$, and $\overleftarrow{\prod_{j=1}^k}$ denotes the matrix product evaluated from index $k$ down to $1$. Consequently, this formulation produces a quantum algorithm achieving {\it non-unitary optimality} for time-independent non-unitary dynamics. The detailed proof can be found in Theorem~\ref{theorem:B2} within Appendix~\ref{section:9-1}.
\begin{theoremmain}
\label{theorem:main:1}
\par There exists a quantum algorithm for simulating $ e^{-At} $ with $A_1\succeq 0$ to within an error tolerance $\varepsilon$. The required query complexity with respect to the associated block-encoding oracle $\mathrm{HAM_A}$ is:
\begin{equation}
    \begin{aligned}
        \label{equ:8}
        \mathcal{O}\left(\frac{\|u_0\|}{\|u(t)\|}\left(A_{\max}t+\frac{\log\varepsilon^{-1}}{\log\log\varepsilon^{-1}}\right)\right),
    \end{aligned}
\end{equation}
where the prefactor $\mathcal{O}\left(\frac{\|u_0\|}{\|u(t)\|}\right)$ stems from the dissipative component and represents the number of repetitions and state-preparation queries to $u_0$. The additional gate complexity scales logarithmically as $\mathcal{O}(\log \varepsilon^{-1})$. The operator $\mathrm{HAM_A}$ is defined as:
\begin{equation}
    \begin{aligned}
        \label{equ:9}
        (\bra{0}_a\otimes I_s)&\mathrm{HAM_A}(\ket{0}_a\otimes I_s)\\
        &=\sum\limits_{m=0}^{N_m-1}|m\rangle\langle m|\otimes\frac{e^{iA_2hm}\tilde{A}_1e^{-iA_2hm}}{A_{\max}}.
    \end{aligned}
\end{equation}
\end{theoremmain}
\noindent The algorithm introduced in Theorem~\ref{theorem:main:1} depends on the block-encoding $\mathrm{HAM}_A$ described by Eq.~(\ref{equ:9}). The primary challenge is achieving the quantum implementation of $e^{iA_2 h m}$ for $m=1,\cdots,N_m$. First, we utilize QSVT to generate a set of base unitary operators $\mathcal{S}=\{e^{iA_2 h 2^j} \mid j = 1,\cdots,\lceil\log N_m\rceil\}$. Doing so requires a query complexity of at most $\mathcal{O}\left(A_{\max}t+\frac{\log\varepsilon^{-1}}{\log\log\varepsilon^{-1}}\right)$. Second, any target operator $e^{iA_2 h m}$ can be synthesized as the product of at most $\mathcal{O}(\lceil\log N_m\rceil)$ elements taken from $\mathcal{S}$. As a result, the gate complexity incurred in building $\mathrm{HAM}_A$ is only $\mathcal{O}(\log N_m)=\mathcal{O}(\log\varepsilon^{-1})$. Overall, the complete procedure demands no more than $\mathcal{O}(\log\varepsilon^{-1})$ auxiliary qubits and preserves non-unitary optimal scaling in this model.

\subsection{Improvements on Phase-Driven APS}
\label{section:3-3}
\par We extend phase-driven APS in three related directions. First, for {\it time-dependent} non-unitary dynamics, we formulate a direct-access model that achieves query complexity $\mathcal{O}\left(\frac{\|u_0\|}{\|u(t)\|}A_{\max}t\frac{\log\varepsilon^{-1}}{\log\log\varepsilon^{-1}}\right)$ while preserving the optimal dependence on queries to $u_0$ and on the repetition count; see Theorem~\ref{theorem:B3} and Appendix~\ref{section:9-2}. This improves prior work by at least a factor of $\log\varepsilon^{-1}$~\cite{Low2019Interaction,Cao2023QuantumSF}. Second, for mildly ill-posed {\it time-independent} problems~\cite{Jin2024IllPosed} with $\lambda_{\min}^-(A)=\lambda_{\min}(A_1)<0$, the same method remains valid and retains optimal scaling under the regime $\lambda_{\min}^-(A)T=\mathcal{O}(1)$; see Appendix~\ref{section:9-3}. Third, for {\it time-independent} systems with inhomogeneous terms~\cite{Jin2025LinearSystems} as in Eq.~(\ref{equ:1}), only the repetition and $u_0$-query factor changes, becoming $\mathcal{O}\left(\frac{\|u_0\|+b_{\max}t}{\|u(t)\|}\right)$ with $b_{\max}=N\max\limits_{i=1}^N|b_i|$; see Appendix~\ref{section:9-4}. These extensions show that the main phase-driven construction is not confined to the clean, homogeneous, well-posed setting.

\section{Fast-Forwarding via Amplitude-Driven APS}
\label{section:4}
Square-root fast-forwarding is the natural performance target on the dissipative side of APS. Unlike Hamiltonian simulation, the simulation of Hermitian operators can evade linear-in-$t$ lower bounds and achieve a square-root dependence on the dissipative scale~\cite{Low2017Forwarding}.
For the {\it time-independent} case, several methods already fast-forward Hermitian operators and normal non-unitary dynamics. These include NDME-based constructions with fast-forwardable Lindbladians~\cite{Shang2025Forwarding}, Gaussian LCHS~\cite{Kharazi2026Sublinear}, and the Kannai transform~\cite{Jin2026Transmutation}.
\subsection{Hermitian Operators}
\label{section:4-1}
\par We begin with time-independent Hermitian operators because they isolate the amplitude block that later appears inside amplitude-driven APS. In this setting, fast-forwarding is fundamentally a Fourier-transform effect: one rewrites the dissipative propagator as a weighted average of unitary evolutions and then truncates that representation in a controlled way~\cite{Low2017Forwarding,Shang2025Forwarding}. This representation makes the later extension to non-unitary dynamics cleaner.
\par The standard target in this setting is a single-run query complexity of $\mathcal{O}(\sqrt{\|L\|t\log\varepsilon^{-1}})$ for simulating $e^{-Lt}$. While current fast-forwarding constructions are often presented together with quantum phase estimation~\cite{Shang2025Forwarding}, APS highlights the Hermitian part itself rather than any single implementation primitive.
\par We start by considering the Fourier transform of $e^{-x^2}$:
\begin{equation}
    \begin{aligned}
        \label{equ:fast:1}
        e^{-x^2}=\frac{1}{2\sqrt{\pi}}\int_{-\infty}^{+\infty}e^{-\frac{\eta^2}{4}}e^{-i\eta x}\d\eta.
    \end{aligned}
\end{equation}
Following a procedure similar to Section~\ref{section:6-1}, we apply spectral decomposition to generalize Eq.~(\ref{equ:fast:1}) to a time-independent positive-definite Hermitian matrix $L$:
\begin{equation}
    \begin{aligned}
        \label{equ:fast:2}
        e^{-Lt}=\frac{1}{2\sqrt{\pi}}\int_{-\infty}^{+\infty}e^{-\frac{\eta^2}{4}}e^{-i\eta \sqrt{Lt}}\d\eta.
    \end{aligned}
\end{equation}
\par As in~\cite{An2023QuantumAF}, we further truncate the integral in Eq.~(\ref{equ:fast:1}) and select a maximum truncation point $\pm M$. Assuming that the block-encoding of $\sqrt{L}$ is known, the final query complexity is related to $M\sqrt{L_{\max}t}$, where $L_{\max}=\|L\|$. Ignoring the final discretization step, we first need to ensure that the norm of the truncated part is less than $\varepsilon$, which requires
\begin{equation}
    \begin{aligned}
        \label{equ:fast:3}
        \frac{1}{2\sqrt{\pi}}\left(\int_{M}^{+\infty}e^{-\frac{\eta^2}{4}}+\int_{-\infty}^{-M}e^{-\frac{\eta^2}{4}}\right)=\frac{1}{\sqrt{\pi}}\int_{M}^{+\infty}e^{-\frac{\eta^2}{4}}<\varepsilon.
    \end{aligned}
\end{equation}
Assuming a truncated version of Eq.~(\ref{equ:fast:2}) with truncation at $\pm M$ (where $M>2$), we use the tail estimate of the Gaussian distribution. We set $M$ by requiring $e^{-\frac{M^2}{4}}<\varepsilon$, resulting in $M=\mathcal{O}\big(\sqrt{\log\varepsilon^{-1}}\big)$. Consequently, the query complexity for constructing the time-independent Hermitian operator $e^{-Lt}$ is
\begin{equation*}
    \begin{aligned}
        \mathcal{Q}=\mathcal{O}\left(\sqrt{L_{\max}t\log\varepsilon^{-1}}\right).
    \end{aligned}
\end{equation*}
To estimate the success probability, we compute the norm of Eq.~(\ref{equ:fast:2}). Note the upper bound $\frac{1}{2\sqrt{\pi}}\int_{-\infty}^{+\infty}e^{-\frac{\eta^2}{4}}\d\eta=1$. Hence, the success probability is estimated as
\begin{equation*}
    \begin{aligned}
        \textbf{Pr}&\sim\|e^{-Lt}\|^2\cdot\left(\frac{1}{2\sqrt{\pi}}\int_{-\infty}^{+\infty}e^{-\frac{\eta^2}{4}}\d\eta\right)^{-2}\ge \left(\frac{\|u(t)\|}{\|u_0\|}\right)^2,
    \end{aligned}
\end{equation*}
where we employ the inequality $\|u(t)\|\le \|e^{-\int_0^tL(s)\d s}\|\cdot\|u_0\|$. Using amplitude amplification, the estimated number of required repetitions is
\begin{equation*}
    \begin{aligned}
        g=\mathcal{O}\left(\frac{\|u_0\|}{\|u(t)\|}\right).
    \end{aligned}
\end{equation*}
Two supplementary discussions are deferred to Appendix~\ref{section:9}: a stochastic extension of Eq.~(\ref{equ:fast:2}) to piecewise time-independent Hermitian operators, and a Fourier-transform argument showing why simple generalizations beyond square-root fast-forwarding do not improve the scaling exponent. We now turn to the main amplitude-driven APS application.
\subsection{Application of Amplitude-Driven APS}
\label{section:4-2}
\par Extending these methods to non-normal non-unitary dynamics remains an ongoing challenge. We overcome this limitation using amplitude-driven APS. This yields the strict decomposition $e^{-At}=e^{-A_1t} \cdot \mathcal{T}e^{-i\int_0^t A_a(s)\d s}$ with the effective interaction $A_a(s)=e^{A_1s}A_2e^{-A_1s}$. This structure allows us to treat the Hermitian operator $e^{-A_1t}$ with proven fast-forwarding methods. Meanwhile, we implement the generally non-unitary operator $\mathcal{T}e^{-i\int_0^t A_a(s)\d s}$ via a Dyson-series-based construction. Combining these approaches yields the following explicit approximation:
\begin{equation}
    \begin{aligned}
        \label{equ:10}
        e^{-At}&\approx\left[\sum\limits_{k=0}^M(-i)^k\left(\frac{h}{N_t}\right)^k\right.\\
        &\quad\left.\sum\limits_{0\le m_1<\cdots<m_k< N_m}e^{-\frac{A_1t}{N_t}}\overleftarrow{\prod_{j=1}^k}e^{\frac{A_1m_jh}{N_t}}A_2e^{-\frac{A_1m_jh}{N_t}}\right]^{N_t}
    \end{aligned}
\end{equation}
where $M=\mathcal{O}\big(\frac{\log\varepsilon^{-1}}{\log\log\varepsilon^{-1}}\big)$ denotes the truncation order. $N_m=\mathcal{O}(M e^{A_{2,\max}t} \max\{A_{1,\max},A_{2,\max}\} A_{2,\max}t \varepsilon^{-1})$ counts the total discrete time intervals, $N_t=\mathcal{O}(A_{2,\max}t)$, and $h=\frac{t}{N_m}$. Writing $r_j=\frac{m_jh}{N_t}$, we implement each dissipative factor $e^{-A_1r_j}$ through the fixed discrete Gaussian approximation
\begin{equation}
    \begin{aligned}
        \label{equ:advection:8}
        e^{-A_1 r_j}
        &\approx \frac{h_a}{2\sqrt{\pi}}\sum\limits_{k=0}^{N_A-1} e^{-\frac{\eta_k^2}{4}} e^{-i\eta_k\sqrt{A_1 r_j}},
    \end{aligned}
\end{equation}
where $\eta_k=-R+\left(k+\frac{1}{2}\right)h_a$ and $h_a=\frac{2R}{N_A}$. This replaces the Gaussian integral from Section~\ref{section:4-1} by a truncated discrete LCU. Hence Eq.~(\ref{equ:10}) is implemented by combining the Dyson-series construction for the interaction factor with a uniform quadrature rule for the short-time dissipative blocks. Since the integration bounds $R$ and the number of quadrature nodes $N_A$ are chosen uniformly for all $r_j\in[0,t]$, the same discretization is used throughout the evolution. The resulting super oracle is summarized in Remark~\ref{lemma:C4}. Based on Eq.~(\ref{equ:10}), this yields a robust square-root fast-forwarding scheme for time-independent non-unitary dynamics. The detailed proof is provided in Appendix~\ref{section:10-2}.

\begin{theoremmain}
\label{theorem:main:2}
\par For simulating $e^{-At}$ assuming $A_1\succeq 0$ within a strictly bounded error tolerance $\varepsilon$, we provide a quantum algorithm whose query complexity to access $\mathrm{HAM}_A$ is
\begin{equation}
    \begin{aligned}
        \label{equ:11}
        \mathcal{O}\left(\frac{\|u_0\|}{\|u(t)\|}\left(\sqrt{A_{1,\max}t\log\varepsilon^{-1}}+A_{2,\max}t\frac{\log\varepsilon^{-1}}{\log\log\varepsilon^{-1}}\right)\right).
    \end{aligned}
\end{equation}
The prefactor $\mathcal{O}\big(\|u_0\|/\|u(t)\|\big)$ reflects the required number of queries to $u_0$ and the overall number of repetitions. The gate complexity scales efficiently as $\mathcal{O}(\log \varepsilon^{-1})$. Here $\mathrm{HAM}_A$ denotes the super oracle from Remark~\ref{lemma:C4}:
\begin{equation}
    \begin{aligned}
        (\bra{0}_a &\otimes I_c \otimes I_s)\textrm{HAM}_A(\ket{0}_a \otimes I_c \otimes I_s)\\
        =& \sum\limits_{j=0}^{N_m-1}\ket{j}\bra{j}_c\otimes\frac{1}{\sum_{k=0}^{N_A-1} e^{-\frac{\eta_k^2}{4}}} \sum_{k=0}^{N_A-1} e^{-\frac{\eta_k^2}{4}} e^{i\eta_k \sqrt{A_1 \frac{m_jh}{N_t}}},
    \end{aligned}
\end{equation}
where the unitary blocks $e^{i\eta_k\sqrt{A_1r_j}}$ are synthesized by the Section~\ref{section:4-1} fast-forwarding routine and combined through an LCU construction after the Gaussian integral has been discretized into a Riemann sum with $N_A$ points.
\end{theoremmain}
\noindent Constructing $\mathrm{HAM}_A$ here follows the unified oracle design given in Remark~\ref{lemma:C4}, so we omit the repetitive circuit details in the main text.

\section{Numerical Simulations}
\label{section:5}

\subsection{1D Advection-Diffusion}
\label{section:5-1}

\begin{figure*}[htbp]
    \centering
    \includegraphics[width=\linewidth]{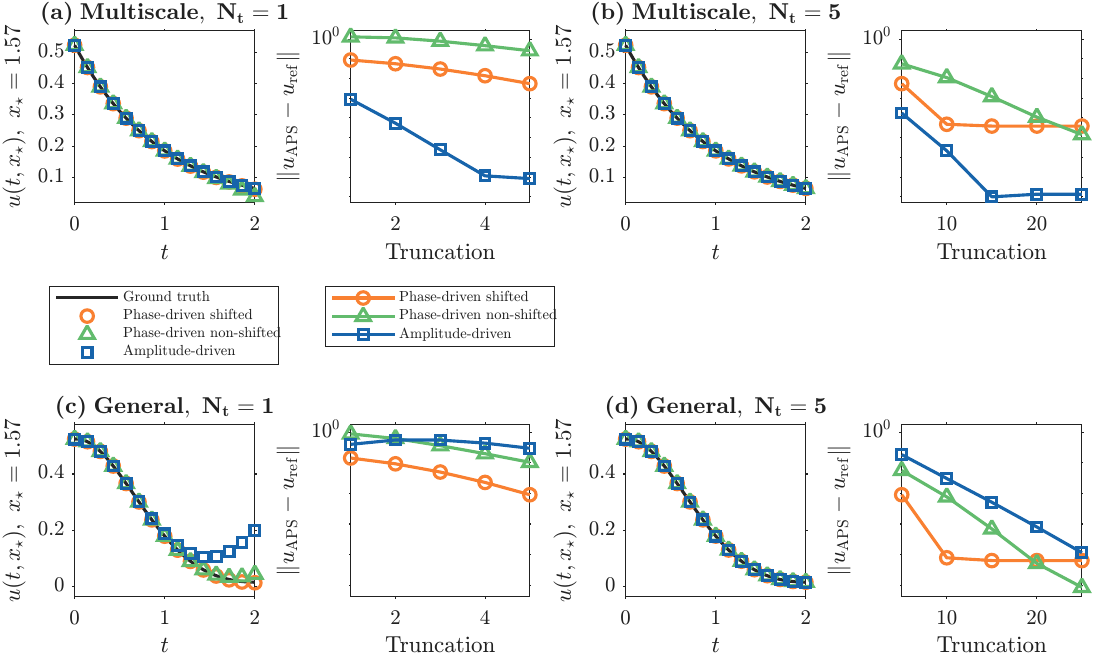}
    \caption{
    \label{figure:advection-benchmark}\textbf{1D advection-diffusion benchmark.} Panels (a)--(d) compare the probe $u(t,x_\star)$ at $x_\star=\pi/2$ for the multiscale regime $(c,\nu,\mu)=(0.1,0.20,1.0)$ and the general regime $(2.0,0.20,0.75)$, each with $N_t=1$ and $N_t=5$. Each trajectory panel compares the ground truth, phase-driven shifted, phase-driven non-shifted, and amplitude-driven constructions, and is followed by its truncation-error panel with $N_m=10000$. In the general regime, shifted phase-driven APS remains stable already at $N_t=1$. In the multiscale regime, amplitude-driven APS has the smallest error and the fastest pre-plateau decay.}
\end{figure*}

\par To connect APS with a PDE-motivated evolution problem, we consider the one-dimensional advection-diffusion-reaction equation
\begin{equation}
    \begin{aligned}
        \label{equ:advection:1}
        \partial_t u+c\partial_x u=\nu\partial_{xx}u-\mu u,
    \end{aligned}
\end{equation}
on a periodic domain $x\in[0,2\pi]$, with constants $c\in\mathbb{R}$ and $\nu,\mu\ge0$. Here the transport term $c\partial_x u$ carries the coherent motion, while diffusion and reaction supply the dissipative part. This makes the model a natural benchmark for testing whether APS separates these two sources of complexity in the way predicted by our theory.

\par Let $x_j=jh$ for $j=0,\ldots,N-1$ with mesh size $h=2\pi/N$, and let $\mathbf{u}(t)=(u(t,x_0),\ldots,u(t,x_{N-1}))^\top$. Using periodic centered differences,
\begin{equation}
    \begin{aligned}
        \label{equ:advection:2}
        (D_1\mathbf{u})_j=\frac{u_{j+1}-u_{j-1}}{2h},\quad (D_2\mathbf{u})_j=\frac{u_{j+1}-2u_j+u_{j-1}}{h^2},
    \end{aligned}
\end{equation}
the semidiscrete system becomes
\begin{equation}
    \begin{aligned}
        \label{equ:advection:3}
        \frac{\d \mathbf{u}(t)}{\d t}=-A\mathbf{u}(t),\quad A=\mu I-\nu D_2+cD_1.
    \end{aligned}
\end{equation}
Because $D_1$ is skew-Hermitian under the periodic centered discretization and $-D_2$ is positive semidefinite, this operator fits our Cartesian decomposition exactly:
\begin{equation}
    \begin{aligned}
        \label{equ:advection:4}
        A=A_1+iA_2,\quad A_1=\mu I-\nu D_2\succeq0,\quad iA_2=cD_1.
    \end{aligned}
\end{equation}
Moreover, the discrete norms satisfy
\begin{equation}
    \begin{aligned}
        \label{equ:advection:5}
        A_{1,\max}\le \mu+\frac{4\nu}{h^2},\quad A_{2,\max}\le \frac{|c|}{h},
    \end{aligned}
\end{equation}
since the eigenvalues of $D_1$ and $D_2$ obey $|\lambda(D_1)|\le h^{-1}$ and $|\lambda(D_2)|\le 4h^{-2}$ on the periodic grid.

\par Under this benchmark, phase-driven APS predicts the query bound
\begin{equation}
    \begin{aligned}
        \label{equ:advection:6}
        \mathcal{Q}_{\mathrm{phase}}=\mathcal{O}\!\left(\frac{\|u_0\|}{\|u(t)\|}\left[\left(\mu+\frac{4\nu}{h^2}\right)t+\frac{\log\varepsilon^{-1}}{\log\log\varepsilon^{-1}}\right]\right),
    \end{aligned}
\end{equation}
while amplitude-driven APS yields
\begin{equation}
    \begin{aligned}
        \label{equ:advection:7}
        \mathcal{Q}_{\mathrm{amp}}&=\mathcal{O}\!\left(\frac{\|u_0\|}{\|u(t)\|}\left[\sqrt{\left(\mu+\frac{4\nu}{h^2}\right)t\log\varepsilon^{-1}}\right.\right.\\
        &\qquad\qquad\qquad\qquad\qquad\left.\left.+\frac{|c|t}{h}\frac{\log\varepsilon^{-1}}{\log\log\varepsilon^{-1}}\right]\right).
    \end{aligned}
\end{equation}
\par In the amplitude-driven computations for Fig.~\ref{figure:advection-benchmark}, we use Eq.~(\ref{equ:advection:8}) with the fixed parameter choice $R=8$ and $N_A=100$. Therefore, this example makes the regime separation explicit: phase-driven APS is the natural route when one wants additive non-unitary optimality, whereas amplitude-driven APS is advantageous when the dissipative scale $\mu+4\nu h^{-2}$ dominates the transport scale $|c|h^{-1}$.

\par Fig.~\ref{figure:advection-benchmark} compares shifted phase-driven APS, its non-shifted phase-driven counterpart, and amplitude-driven APS in the multiscale and general regimes at $N_t=1$ and $N_t=5$. The phase-driven non-shifted baseline is obtained by removing the shifted technique from the phase-driven construction; without that shift, it may require time segmentation, so we compare it for both $N_t=1$ and $N_t=5$. By contrast, the shifted phase-driven construction is unconditionally stable and, in the general regime, is the only trajectory that stays close to the ground truth already at small $N_t$. The multiscale row highlights the complementary advantage of amplitude-driven APS, whose truncation error is both smaller and faster-decaying before saturation. The late plateaus mainly come from fixed choices such as $N_m$, not from a change in the trend.

\subsection{Bloch-Type Relaxation}
\label{section:5-2}

\begin{figure*}[htbp]
    \centering
    \includegraphics[width=\linewidth]{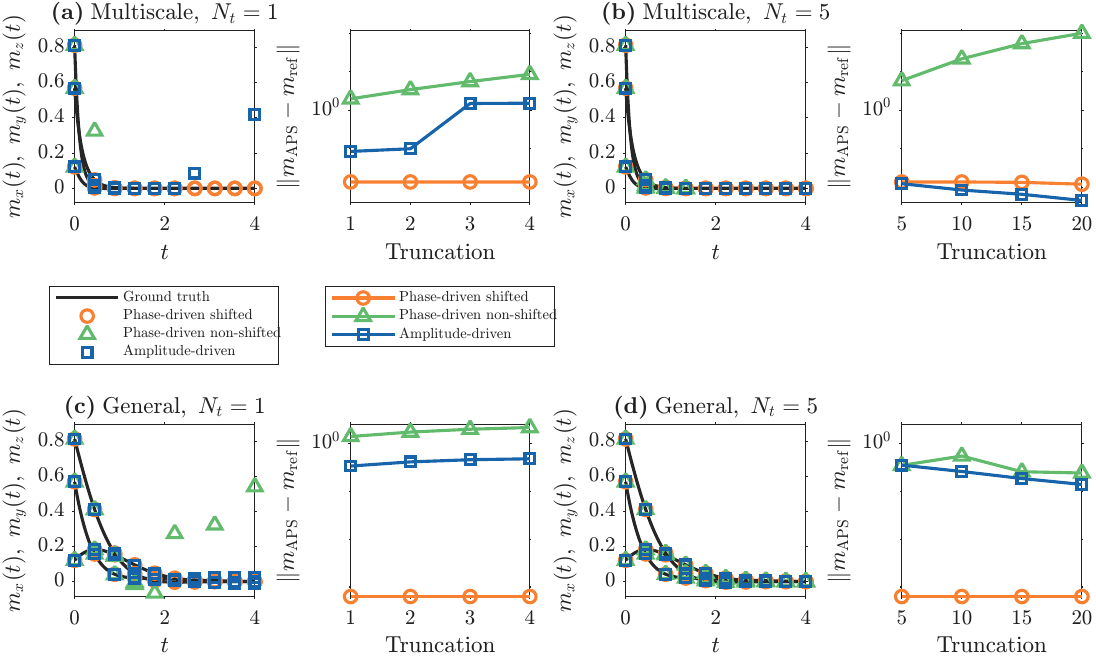}
    \caption{
    \label{figure:bloch-benchmark}\textbf{Bloch-type relaxation benchmark.} Panels (a)--(d) show the multiscale regime $\Omega=(0.50,0.80,1.10)^\top$, $T_1=0.2$, $T_2=0.1$ and the general regime $\Omega=(0.60,0.80,1.00)^\top$, $T_1=T_2=0.6$, each with $N_t=1$ and $N_t=5$; each trajectory panel is followed by its truncation-error panel with $N_m=5000$. In the general regime, shifted phase-driven APS is already the most stable at $N_t=1$. In the multiscale regime, amplitude-driven APS again shows the smallest and fastest-decaying pre-plateau error.}
\end{figure*}

\par As a complementary finite-dimensional benchmark to the advection-diffusion PDE, we consider the homogeneous Bloch-type relaxation equation
\begin{equation}
    \begin{aligned}
        \label{equ:bloch:1}
        \frac{\d \mathbf{m}(t)}{\d t}=\Omega\times \mathbf{m}(t)-\Gamma\mathbf{m}(t),
    \end{aligned}
\end{equation}
where $\mathbf{m}(t)=(m_x,m_y,m_z)^\top$ is the magnetization vector, $\Omega=(\Omega_x,\Omega_y,\Omega_z)^\top$ is the precession field, and $\Gamma=\mathrm{diag}(T_2^{-1},T_2^{-1},T_1^{-1})\succeq0$ is the relaxation matrix. This model combines coherent precession and anisotropic dissipation in a standard ODE setting from magnetic-resonance and open-system dynamics.

\par Writing the cross product as a skew-symmetric matrix action,
\begin{equation}
    \begin{aligned}
        \label{equ:bloch:2}
        K_{\Omega}=\begin{bmatrix}
            0 & -\Omega_z & \Omega_y\\
            \Omega_z & 0 & -\Omega_x\\
            -\Omega_y & \Omega_x & 0
        \end{bmatrix},\quad K_{\Omega}\mathbf{m}=\Omega\times\mathbf{m},
    \end{aligned}
\end{equation}
the dynamics takes the form
\begin{equation}
    \begin{aligned}
        \label{equ:bloch:3}
        \frac{\d \mathbf{m}(t)}{\d t}=-A\mathbf{m}(t),\qquad A=\Gamma-K_{\Omega}.
    \end{aligned}
\end{equation}
This again fits the APS decomposition exactly:
\begin{equation}
    \begin{aligned}
        \label{equ:bloch:4}
        A=A_1+iA_2,\quad A_1=\Gamma\succeq0,\quad iA_2=-K_{\Omega},
    \end{aligned}
\end{equation}
so that $A_2=iK_{\Omega}$ is Hermitian. The corresponding scales satisfy
\begin{equation}
    \begin{aligned}
        \label{equ:bloch:5}
        A_{1,\max}=\|\Gamma\|=\max\{T_1^{-1},T_2^{-1}\},\quad A_{2,\max}=\|A_2\|=\|\Omega\|_2.
    \end{aligned}
\end{equation}
Hence phase-driven APS predicts
\begin{equation}
    \begin{aligned}
        \label{equ:bloch:6}
        \mathcal{Q}_{\mathrm{phase}}=\mathcal{O}\!\left(\frac{\|\mathbf{m}_0\|}{\|\mathbf{m}(t)\|}\left[\max\{T_1^{-1},T_2^{-1}\}t+\frac{\log\varepsilon^{-1}}{\log\log\varepsilon^{-1}}\right]\right),
    \end{aligned}
\end{equation}
while amplitude-driven APS yields
\begin{equation}
    \begin{aligned}
        \label{equ:bloch:7}
        \mathcal{Q}_{\mathrm{amp}}&=\mathcal{O}\!\left(\frac{\|\mathbf{m}_0\|}{\|\mathbf{m}(t)\|}\left[\sqrt{\max\{T_1^{-1},T_2^{-1}\}t\log\varepsilon^{-1}}\right.\right.\\
        &\qquad\qquad\qquad\qquad\qquad\left.\left.+\|\Omega\|_2 t\frac{\log\varepsilon^{-1}}{\log\log\varepsilon^{-1}}\right]\right).
    \end{aligned}
\end{equation}
This example isolates the same crossover in a finite-dimensional setting: when relaxation dominates precession, amplitude-driven APS inherits the square-root improvement from the dissipative block, whereas phase-driven APS remains the natural choice when one prioritizes additive non-unitary optimality.

\par In the numerical test, we fix the initial magnetization to $\mathbf{m}_0=(1,0.15,0.70)^\top/\|(1,0.15,0.70)^\top\|$, evolve over $t\in[0,4]$, and use the common discretization $N_m=5000$. We compare two segmentation choices, $N_t=1$ and $N_t=5$, under the matched phase budget $M_{\mathrm{phase}}=N_tM$. The multiscale row uses $T_2=0.1$ and $T_1=0.2$, so the dissipative scale clearly dominates the precession scale $\|\Omega\|_2$, while the general row with $T_1=T_2=0.6$ is more balanced.

\par Fig.~\ref{figure:bloch-benchmark} repeats the PDE message in a finite-dimensional setting. In the general row, shifted phase-driven APS is the only method that stays close to the ground truth at $N_t=1$. In the multiscale row, amplitude-driven APS starts lower and decays faster before saturation. As in Fig.~\ref{figure:advection-benchmark}, the final plateaus mainly come from fixed choices such as $N_m$.

\section{Connections to Existing Methods}
\label{section:6}

\par APS is also useful as an interpretive tool: it isolates which parts of existing non-unitary algorithms come from Hamiltonian-simulation primitives and which come from genuinely dissipative structure. We therefore use this Section~to relate APS to two representative methods, LCHS and NDME, and to identify the bottlenecks that remain after this reorganization.

\subsection{Linear Combination of Hamiltonian Simulation (LCHS)}
\label{section:6-1}
\par LCHS is a clear Hamiltonian-style example of APS. Rather than treating it as an unrelated method, we show that LCHS is phase-driven APS plus a Fourier-approximation step.
\subsubsection{LCHS = Phase-Driven APS + Fourier Transform}
\par LCHS can be derived by combining phase-driven APS with a Fourier representation of $e^{-x}$ on $x\ge0$. Examples include $e^{-|x|}$ in the original LCHS~\cite{An2023QuantumAF}, bump functions~\cite{Huang2025Fourier}, and Gaussian kernels~\cite{Low2025Optimal}, but all such variants rely on the same basic identity:
\begin{equation}
    \begin{aligned}
        \label{equ:LCHS:1}
        e^{-x}=\int_{-\infty}^{+\infty}\gamma(\eta)e^{-i\eta x}\d\eta,\quad x\ge0.
    \end{aligned}
\end{equation}
Applying spectral decomposition to a positive semidefinite Hermitian matrix $L$ generalizes Eq.~(\ref{equ:LCHS:1}) directly to matrices:
\begin{equation}
    \begin{aligned}
        \label{equ:LCHS:2}
        e^{-Lt}=\int_{-\infty}^{+\infty}\gamma(\eta)e^{-i\eta Lt}\d\eta,\quad t\in[0,T],\quad L\succeq 0.
    \end{aligned}
\end{equation}
For time-dependent Hermitian operators, this extension requires the analyticity conditions of An et al.~\cite{An2023Optimal}. Under those conditions, one obtains:
\begin{equation}
    \begin{aligned}
        \label{equ:LCHS:3}
        \mathcal{T}e^{-\int_0^tL(s)\d s}=\int_{-\infty}^{+\infty}\gamma(\eta)\mathcal{T}e^{-i\int_0^t\eta L(s)\d s}\d\eta,\quad t\in[0,T].
    \end{aligned}
\end{equation}
Finally, we can derive the general LCHS expression for a non-unitary operator $\mathcal{T}e^{-\int_0^tA(s)\d s}$. This is obtained by inserting Eq.~(\ref{equ:LCHS:3}) into the phase-driven APS from Theorem~\ref{theorem:A1}:
\begin{widetext}
\begin{equation}
    \begin{aligned}
        \label{equ:LCHS:4}
        \mathcal{T}e^{-\int_0^tA(s)\d s}
        &=\mathcal{T}e^{-i\int_0^tA_2(s)\d s}\cdot\mathcal{T}e^{-\int_0^t(\mathcal{T}e^{-i\int_0^s A_2(r)\d r})^\dagger A_1(s)\mathcal{T}e^{-i\int_0^s A_2(r)\d r}\d s}\\
        &=\int_{-\infty}^{+\infty}\gamma(\eta)\mathcal{T}e^{-i\int_0^tA_2(s)\d s}\cdot\mathcal{T}e^{-i\eta\int_0^t(\mathcal{T}e^{-i\int_0^s A_2(r)\d r})^\dagger A_1(s)\mathcal{T}e^{-i\int_0^s A_2(r)\d r}\d s}\d\eta\\
        &=\int_{-\infty}^{+\infty}\gamma(\eta)\mathcal{T}e^{-i\int_0^t(\eta A_1(s)+A_2(s))\d s}\d\eta.
    \end{aligned}
\end{equation}
\end{widetext}
Thus, LCHS is precisely phase-driven APS plus a Fourier-transform primitive. APS supplies the decomposition, while the Fourier transform supplies the approximation mechanism.

\subsubsection{Non-Optimality of LCHS}
\par LCHS remains near-optimal but not strictly additive in the non-unitary sense. The obstruction already appears in the scalar surrogate $e^{-\lambda t}=\int\gamma(\eta)e^{-i\eta\lambda t}\d\eta$: to reach error $\varepsilon$, the Fourier domain must be truncated at tolerance-dependent cutoffs. If those cutoffs stayed bounded as $\varepsilon\to0$, then $\gamma(\eta)$ would have compact support, and a Paley-Wiener argument would contradict the exponential decay of $e^{-\lambda t}$ along the negative real axis. Thus the multiplicative tolerance overhead is inherited from Fourier truncation rather than from the APS decomposition itself. A detailed proof and the corresponding numerical verification are given in Appendix~\ref{section:11-1}; see Fig.~\ref{figure:lchs-scaling}.
\subsection{Non-Diagonal Density Matrix Encoding (NDME)}
\label{section:6-2}

\subsubsection{Interaction Picture for Lindbladians}

\par We now switch to the Lindbladian side. The point is to show that phase-driven APS naturally induces the interaction picture for Lindbladians:
\begin{equation}
    \begin{aligned}
        \label{equ:lindbladian:1}
        \frac{\d \rho(t)}{\d t}&=-i[H(t),\rho(t)]\\
        &\quad+\sum\limits_{i=1}^K\left(F_i(t)\rho(t) F_i^\dagger(t)-\frac{1}{2}\left\{\rho(t),F_i^\dagger(t) F_i(t)\right\}\right),
    \end{aligned}
\end{equation}
where $\rho(t)=\sum\limits_{i=1}^N\sum\limits_{j=1}^N\rho_{ij}(t)\ket{i}\bra{j}$ is the density matrix, $H(t)$ is the Hamiltonian, and $F_i(t)$ are dissipative/jump operators. Following the same construction as in phase-driven APS, we define $\mathcal{U}_p(t)=\mathcal{T}e^{-i\int_0^tH(s)\d s}$, which satisfies the ODE $\frac{\d \mathcal{U}_p(t)}{\d t}=-iH(t)\mathcal{U}_p(t)$. We then define the interaction-picture state $\rho_p(t)=\mathcal{U}_p^\dagger(t)\rho(t)\mathcal{U}_p(t)$. Using the same calculation pattern as in Theorem~\ref{theorem:A1}, we obtain the following evolution equation:
\begin{widetext}
    \begin{align}
        \nonumber\frac{\d \rho_p(t)}{\d t}
        \nonumber&=\frac{\d\mathcal{U}_p^\dagger(t)}{\d t}\rho(t)\mathcal{U}_p(t)+\mathcal{U}_p^\dagger(t)\frac{\d\rho(t)}{\d t}\mathcal{U}_p(t)+\mathcal{U}_p^\dagger(t)\rho(t)\frac{\d\mathcal{U}_p(t)}{\d t}\\
        \nonumber&=\sum\limits_{i=1}^K\left(F_{pi}(t)\rho_p(t) F_{pi}^\dagger(t)-\frac{1}{2}\rho_p(t)F_{pi}^\dagger(t) F_{pi}(t)-\frac{1}{2}F_{pi}^\dagger(t) F_{pi}(t)\rho_p(t)\right)\\
        \label{equ:lindbladian:2}&=\sum\limits_{i=1}^K\left(F_{pi}(t)\rho_p(t) F_{pi}^\dagger(t)-\frac{1}{2}\left\{\rho_p(t),F_{pi}^\dagger(t) F_{pi}(t)\right\}\right),
    \end{align}
\end{widetext}
where $H_d(t)=\mathcal{U}_p^\dagger(t)H(t)\mathcal{U}_p(t)$ and $F_{pi}(t)=\mathcal{U}_p^\dagger(t)F_i(t)\mathcal{U}_p(t)$. Eq.~(\ref{equ:lindbladian:2}) shows that $\rho_p(t)$ satisfies a Lindblad equation with only time-dependent dissipative/jump terms, and it preserves the trace, i.e., $\text{tr}[\rho_p(t)] = \text{tr}[\rho(t)]$. Based on this, we recover $\rho(t)$ from $\rho_p(t)$:
\begin{equation}
    \begin{aligned}
        \label{equ:lindbladian:3}
        \rho(t)=\mathcal{U}_p(t)\rho_p(t)\mathcal{U}_p^\dagger(t).
    \end{aligned}
\end{equation}
This transforms the evolution of $\rho(t)$ into unitary conjugations together with a Lindblad equation containing only dissipative/jump terms, namely the interaction picture for Lindbladians. In APS language, it shifts the main difficulty from the original non-unitary operator to a purely dissipative block. The remaining weak point is again time dependence: nearly optimal algorithms are known mainly in the time-independent case~\cite{Shang2025ExpLindbladian}.

\begin{figure*}[htbp]
    \centering
    \includegraphics[width=0.75\linewidth]{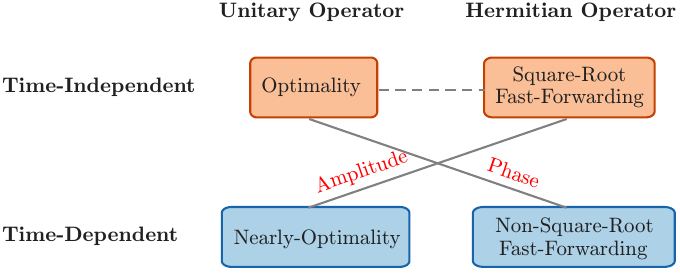}
    \caption{
    \label{figure:limitations-overview}\textbf{Schematic comparison of the current frontiers relevant to APS.} The diagram contrasts unitary versus Hermitian operators and time-independent versus time-dependent settings, indicating where phase-driven and amplitude-driven APS inherit current strengths and where the main bottlenecks remain.}
\end{figure*}

\subsubsection{Emergence of NDME = Phase-Driven APS + Interaction Picture for Lindbladians}

\par NDME can then be viewed as the block-matrix realization of the same APS-plus-interaction-picture route. To expose only the essential structure in the main text, we define the interaction-picture block variable
\begin{equation}
    \begin{aligned}
        \label{equ:lindbladian:5}
        \rho_p(t)=\begin{bmatrix}
            \tilde{\rho}_p(t) & u_p(t)\\
            u_p^\dagger(t) & 1
        \end{bmatrix},
    \end{aligned}
\end{equation}
where $u_p(t)$ is the phase-driven APS state and $\tilde{\rho}_p(t)$ is an auxiliary matrix with $\mathrm{tr}[\tilde{\rho}_p(t)]=0$. For the corresponding dissipative block operator $F_p(t)$, this variable satisfies
\begin{equation}
    \begin{aligned}
        \label{equ:lindbladian:6}
        \frac{\d\rho_p(t)}{\d t}
        &=F_p(t)\rho_p(t)F_p^\dagger(t)
        -\frac{1}{2}\{\rho_p(t),F_p^\dagger(t)F_p(t)\}.
    \end{aligned}
\end{equation}
This is the NDME-type dissipative equation after the coherent part has been peeled off by phase-driven APS. In this form, APS makes the bottleneck transparent: the remaining difficulty lies in simulating the time-dependent dissipative block rather than in the block encoding itself. The full block-matrix embedding and the inverse transformation back to the canonical NDME master equation are deferred to Appendix~\ref{section:11-2}.

\section{Summary and Outlook}
\label{section:7}
\subsection{Summary}
\label{section:7-1}
In this work, we introduce {\it Amplitude-Phase Separation} (APS) to separate coherent and dissipative structure in time-dependent non-unitary evolution. For time-independent simulation, phase-driven APS plus the shifted Dyson series achieves additive non-unitary-optimal scaling, meaning additive rather than multiplicative tolerance overhead, while amplitude-driven APS yields square-root fast-forwarding with respect to the dissipative scale in a broader non-unitary regime.

\par These gains are complementary rather than simultaneous. APS also clarifies that LCHS is phase-driven APS plus Fourier approximation, while NDME is phase-driven APS plus a Lindbladian interaction picture. The benchmarks confirm the predicted crossover between phase-driven and amplitude-driven advantages in both PDE and ODE settings. The main remaining weakness lies in the time-dependent dissipative components.

\subsection{Outlook}
\label{section:7-2}
The next targets are the time-dependent Hermitian and interaction-picture components behind Fig.~\ref{figure:limitations-overview}. On the phase-driven side, fast-forwardable time-dependent Hermitian simulation would broaden the range of dissipative speedups. On the amplitude-driven side, the interaction factor still carries a multiplicative $\log\varepsilon^{-1}$ overhead~\cite{Low2019Interaction,Cao2023QuantumSF}. Existing Lindbladian fast-forwarding results~\cite{Shang2025ExpLindbladian,Shang2025Forwarding} and our piecewise time-independent construction in Appendix~\ref{section:10} are natural starting points.

\par More broadly, APS may also help beyond simulation itself. In particular, the same separation may extend to quantum linear systems $Ax=b$~\cite{Harrow2009QuantumAF,Subasi2019PRL,Berry2022PRXQuantum} and related scientific-computing tasks in which dissipative and oscillatory effects coexist, while keeping the coherent and dissipative sources of complexity analytically distinct.

\section*{Acknowledgments}
SJ acknowledges the support of the NSFC grant No. 12341104, the Shanghai Pilot Program for Basic Research, the Science and Technology Commission of Shanghai Municipality (STCSM) grant no. 24LZ1401200, the Shanghai Jiao Tong University 2030 Initiative, and the Fundamental Research Funds for the Central Universities. QH thanks Xiaoyang He in Shanghai Jiao Tong University for useful suggestions on proofs.

\section*{Data Availability}
Data and code used in this research will be publicly available upon acceptance.

\bibliography{reference}
\appendix
\section{Additional Discussion of the Shifted Dyson Series}
\label{section:8}
\subsection{Detailed Proof for Theorem~\ref{theorem:B1}}
\label{section:8-1}
\par In this section, we present the Theorem~for approximating Hermitian operators using Riemann sums of the Dyson series.
\begin{lemma}
\label{lemma:B2}
\par Let $L(s)$ be a Hermitian matrix satisfying $L(s)\succeq 0$ for all $s\in[0,T]$, and define $L_{\max}=\sup\limits_{s\in[0, t]}\|L(s)\|$, $\tau=tL_{\max}$, $\hat{L}_{\max}=\sup\limits_{s\in[0, t]}\left\|\left.\frac{\d L(t)}{\d t}\right|_{t=s}\right\|$. For any error tolerance $\varepsilon > 0$, if the number of discrete intervals $N_m$ is chosen such that
\begin{equation}
    \begin{aligned}
        \label{equ:dyson:9}
        N_m\ge \frac{M(\hat{L}_{\max}t+L_{\max}(\tau+1))}{\varepsilon}.
    \end{aligned}
\end{equation}
Then the sum $\sum\limits_{k=0}^M(-1)^kP_k$ can be approximated by $\sum\limits_{k=0}^M(-1)^kQ_k$ within an error bound of $\varepsilon$:
\begin{equation*}
    \begin{aligned}
        \left\|\sum\limits_{k=0}^M(-1)^kQ_k-\sum\limits_{k=0}^M(-1)^kP_k\right\|\le\varepsilon,
    \end{aligned}
\end{equation*}
where $Q_k$, which does not contain the time-ordering operator $\mathcal{T}$, is defined as
\begin{equation*}
    \begin{aligned}
        Q_k=e^{-\tau}h^k\sum\limits_{0\le m_1<\cdots<m_k< N_m}\tilde{L}(m_kh)\cdots\tilde{L}(m_1h),
    \end{aligned}
\end{equation*}
in which $\tilde{L}(t)=L(t)-L_{\max}I$ and $m_1,\cdots,m_k$ are indices in $0,\cdots,N_m$.
\begin{proof}
\par We can rewrite $P_k$ in the following form:
\begin{equation}
    \begin{aligned}
        \label{equ:dyson:10}
        P_k=e^{-\tau}\int_0^t\d t_k\tilde{L}(t_k)
        \cdots\int_0^{t_2}\d t_1\tilde{L}(t_1).
    \end{aligned}
\end{equation}
The approximation of $P_k$ consists of two parts: the first part handles the time dependence of $L(t)$, and the second part deals with the integration $\int_0^{t_{j+1}}$.\\
    $\bullet$ {\bf Step 1: the time dependence of $L(t)$.} Based on the definition of $m_k$, we have $t_k=m_k h+{\Delta_k}$, where the remainder clearly satisfies ${\Delta_k}<h$. We then define the following $R_k$, {which approximates $P_k$ by replacing $t_k$ with $m_kh$}:
\begin{equation}
    \begin{aligned}
        \label{equ:dyson:11}
        R_k=e^{-\tau}\int_0^t\d t_k\tilde{L}(m_kh)
        \cdots\int_0^{t_2}\d t_1\tilde{L}(m_1h).
    \end{aligned}
\end{equation}
Since the main difference clearly lies between $\tilde{L}(t_j)$ and $\tilde{L}(m_j h)$, we can derive an upper bound for this difference as follows:
\begin{equation*}
    \begin{aligned}
        \|\tilde{L}(t_j)-\tilde{L}(m_jh)\|&=\|L(m_j{h}+{\Delta_j})-L(m_jh)\|\\
        &\le {\Delta_j}\hat{L}_{\max}<h\hat{L}_{\max},
    \end{aligned}
\end{equation*}
where $\hat{L}_{\max}$ is the Lipschitz coefficient, defined as $\sup\limits_{s\in[0, t]} \left\|\frac{d L(s)}{ds} \right\|$. Therefore, using this upper bound estimate, we can obtain the error between the sum of $P_k$ in Eq.~(\ref{equ:dyson:10}) and the sum of $R_k$ in Eq.~(\ref{equ:dyson:11}) as follows by using the fact $\|\tilde{L}(t)\|=\|L(t)-L_{\max} I\| \le L_{\max}$ and setting $t_{k+1}=t$:
\begin{widetext}
    \begin{align}
        \nonumber&\left\|\sum\limits_{k=0}^M(-1)^kP_k-\sum\limits_{k=0}^M(-1)^kR_k\right\|\le\sum\limits_{k=1}^M\|P_k-R_k\|\\
        \nonumber\le& e^{-\tau}\sum\limits_{k=1}^M\left(\prod_{p=k}^2\int_0^{t_{p+1}}\d t_p\|\tilde{L}(t_p)\|\cdot\int_0^{t_2}\d t_1\|\tilde{L}(t_1)-\tilde{L}(m_1h)\|\right.\\
        \nonumber&\qquad\qquad+\sum\limits_{j=2}^{k-1}\big(\prod_{p=k}^{j+1}\int_0^{t_{p+1}}\d t_p\|\tilde{L}(t_p)\|\cdot\int_0^{t_{j+1}}\d t_j\|\tilde{L}(t_j)-\tilde{L}(m_jh)\|\cdot\prod_{q=j-1}^1\int_0^{t_{q+1}}\d t_{q}\|\tilde{L}(m_{q}h)\|\big)\\
        \nonumber&\qquad\qquad+\left.\int_0^{t_{k+1}}\d t_k\|\tilde{L}(t_k)-\tilde{L}(m_kh)\|\cdot\prod_{q=k-1}^1\int_0^{t_{q+1}}\d t_{q}\|\tilde{L}(m_{q}h)\|\right)\\
        \nonumber\le& h\hat{L}_{\max}e^{-\tau}\sum\limits_{k=1}^ML_{\max}^{k-1}\sum\limits_{j=1}^k\left(\int_0^t\d t_k\cdots\int_0^{t_{j+1}}\d t_j\cdots\int_0^{t_2}\d t_1\right)\\
        \label{equ:dyson:12}=&ht\hat{L}_{\max}\sum\limits_{k=1}^M\frac{e^{-\tau}\tau^{k-1}}{(k-1)!}<ht\hat{L}_{\max}.
    \end{align}
\end{widetext}
Note that the ordered product $\prod$ is evaluated from larger to smaller indices, and the matrix with a larger (smaller) index is placed on the left (right). This satisfies the requirement of the time-ordering operator.\\
$\bullet$ {\bf Step 2: the integration $\int_0^{t_{j+1}}$.} We need to compute the error introduced by the integration region. First, we can express $Q_k$ in the following split format:
\begin{equation*}
    \begin{aligned}
        Q_k&=e^{-\tau}h^k\sum\limits_{m_k=0}^{N_m-1}\tilde{L}(m_kh)
        \cdots\sum\limits_{m_1=0}^{m_2-1}\tilde{L}(m_1h)
    \end{aligned}
\end{equation*}
and in the form according to the definition of the integral:
\begin{equation}
    \begin{aligned}
        \label{equ:dyson:13}
        Q_k&=e^{-\tau}\int_0^t\d t_k\tilde{L}(m_kh)
        \cdots\int_0^{m_2h}\d t_1\tilde{L}(m_1h),
    \end{aligned}
\end{equation}
where evidently there is at most an $O(h)$ discrepancy in each integration region. We first consider the difference over any single time segment, which yields the following inequality:
\begin{equation*}
    \begin{aligned}
        \int_{m_{j+1}h}^{t_{j+1}}\|\tilde{L}(m_jh)\|\d t_j\le L_{\max}\Delta_{j+1}< L_{\max}h,
    \end{aligned}
\end{equation*}
where we use the property that $\|\tilde{L}(t)\| \le L_{\max}$. By setting $t_{k+1}=t$ and $m_{k+1}=\lfloor\frac{t}{h}\rfloor$, the difference between the sums of $R_k$ in Eq.~(\ref{equ:dyson:11}) and $Q_k$ in Eq.~(\ref{equ:dyson:13}) is calculated as follows:
\begin{widetext}
    \begin{align}
        \nonumber&\left\|\sum\limits_{k=0}^M(-1)^kR_k-\sum\limits_{k=0}^M(-1)^kQ_k\right\|\le\sum\limits_{k=1}^M\left\|R_k-Q_k\right\|\\
        \nonumber\le& e^{-\tau}\sum\limits_{k=1}^M\left(\prod_{p=k}^2\int_0^{t_{p+1}}\d t_p\|\tilde{L}(m_ph)\|\cdot\int_{m_2h}^{t_2}\d t_1\|\tilde{L}(m_1h)\|\right.\\
        \nonumber&\qquad\qquad+\sum\limits_{j=2}^{k-1}\big(\prod_{p=k}^{j+1}\int_0^{t_{p+1}}\d t_p\|\tilde{L}(m_ph)\|\cdot\int_{m_{j+1}h}^{t_{j+1}}\d t_j\|\tilde{L}(m_jh)\|\cdot\prod_{q=j-1}^1\int_0^{m_{q+1}h}\d t_{q}\|\tilde{L}(m_{q}h)\|\big)\\
        \nonumber&\qquad\qquad+\left.\int_{m_{k+1}h}^{t_{k+1}}\d t_k\|\tilde{L}(m_kh)\|\cdot\prod_{q=k-1}^1\int_0^{m_{q+1}h}\d t_{q}\|\tilde{L}(m_{q}h)\|\right)\\
        \nonumber\le& e^{-\tau}\sum\limits_{k=1}^M\left(\prod_{p=k}^2\int_0^{t_{p+1}}\d t_p\|\tilde{L}(m_ph)\|\cdot\int_{m_2h}^{t_2}\d t_1\|\tilde{L}(m_1h)\|\right.\\
        \nonumber&\qquad\qquad+\sum\limits_{j=2}^{k-1}\big(\prod_{p=k}^{j+1}\int_0^{t_{p+1}}\d t_p\|\tilde{L}(m_ph)\|\cdot\int_{m_{j+1}h}^{t_{j+1}}\d t_j\|\tilde{L}(m_jh)\|\cdot\int_0^{t_{j+1}}\d t_{j-1}\|\tilde{L}(m_{j-1}h)\|\cdot\prod_{q=j-2}^1\int_0^{m_{q+1}h}\d t_{q}\|\tilde{L}(m_{q}h)\|\big)\\
        \nonumber&\qquad\qquad+\left.\int_{m_{k+1}h}^{t_{k+1}}\d t_k\|\tilde{L}(m_kh)\|\cdot\int_0^{t_{k+1}}\d t_{k-1}\|\tilde{L}(m_{k-1}h)\|\cdot\prod_{q=k-2}^1\int_0^{m_{q+1}h}\d t_{q}\|\tilde{L}(m_{q}h)\|\right)\\
        \nonumber\le& e^{-\tau}\sum\limits_{k=1}^M hL_{\max}^k\left(\prod_{p=k}^3\int_0^{t_{p+1}}\d t_p\cdot\int_0^{t_3}\d t_2+\int_{0}^{t_{k+1}}\d t_{k-1}\cdot\prod_{q=k-2}^1\int_0^{m_{q+1}h}\d t_{q}\right.\\
        \nonumber&\qquad\qquad\qquad\quad\ \left.+\sum\limits_{j=2}^{k-1}\big(\prod_{p=k}^{j}\int_0^{t_{p+1}}\d t_p\cdot\int_0^{t_{j+2}}\d t_{j+1}\cdot\int_0^{t_{j+1}}\d t_{j-1}\cdot\prod_{q=j-2}^1\int_0^{m_{q+1}h}\d t_{q}\big)\right)\\
        \label{equ:dyson:14}\le& he^{-\tau}\sum\limits_{k=1}^ML_{\max}^k\frac{kt^{k-1}}{(k-1)!}\le hL_{\max}(\tau+1).
    \end{align}
\end{widetext}
\par Combining the results from Eq.~(\ref{equ:dyson:12}) and (\ref{equ:dyson:14}), one can obtain an upper bound for the difference between the sums of $P_k$ and $Q_k$:
\begin{equation*}
    \begin{aligned}
        \left\|\sum\limits_{k=0}^M(-1)^kP_k-\sum\limits_{k=0}^M(-1)^kQ_k\right\|
        \le& h(\hat{L}_{\max}t+L_{\max}(\tau+1)).
    \end{aligned}
\end{equation*}
To ensure this upper bound is controlled by the error tolerance $\varepsilon$, $N_m$ must satisfy the following condition:
\begin{equation*}
    \begin{aligned}
        N_m\ge \frac{M(\hat{L}_{\max}t+L_{\max}(\tau+1))}{\varepsilon}.
    \end{aligned}
\end{equation*}
This completes the proof.
\end{proof}
\end{lemma}

\subsection{Necessity and Insufficiency of the Dissipative Process for the Shift Technique}
\label{section:8-2}
\par A natural question arises: why has the shift technique not been used for Hamiltonian simulation, while it has been studied for Hermitian operators? In this section, we use theoretical arguments to address this issue by illustrating the necessity of dissipative dynamics and why dissipation alone is insufficient for the shift technique.
\par Although we provide a series approximation for the Hermitian operator in Theorem~\ref{theorem:B1}, designing the actual quantum algorithm requires encoding $\tilde{L}(t)$, which inherently introduces errors. We want to ensure that these errors do not amplify exponentially within the quantum circuit. Therefore, as highlighted in Remark~\ref{remark:B3}, the norm of the series $c$ must not be too large. At the very least, it should remain independent of $t$.
\par For necessity, consider the special case $L(t)=A$. Let the shift point be $\beta\in \mathbb{C}$. Then we have
\begin{equation*}
    \begin{aligned}
        e^{-At}&=e^{-\beta t}e^{-(A-\beta I)t}\\
        &=e^{-\beta t}\sum\limits_{k=0}^\infty \frac{1}{k!}\int_0^t\cdots\int_0^t(\beta I-A)^k\d^k t,
    \end{aligned}
\end{equation*}
which means that the degree of amplitude amplification can be characterized as
\begin{equation}
    \begin{aligned}
        \label{equ:dyson:8}
        c=&e^{-\text{Re}[\beta] t}\sum\limits_{k=0}^\infty \frac{1}{k!}\int_0^t\cdots\int_0^t\|\beta I-A\|^k\d^k t\\
        =&e^{(-\text{Re}[\beta]+\|\beta I-A\|)t}\ge e^{(-\text{Re}[\beta]+\rho(\beta I-A))t},
    \end{aligned}
\end{equation}
where we use the inequality $\|\beta I-A\|\ge\rho(\beta I-A)$. If $-\text{Re}[\beta] + \rho(\beta I - A) > 0$, the norm $c$ will inevitably increase exponentially with $t$. Moreover, computing the spectral radius $\rho(\beta I-A)$ is equivalent to solving an optimization problem for a fixed $\beta$: $\max\limits_{\lambda_i \in \Lambda(A)} |\beta-\lambda_i|$, where $\Lambda(A)$ denotes the set of eigenvalues of $A$. Therefore, to keep the norm $c$ constant, the following condition is necessary:
\begin{equation}
    \begin{aligned}
        -\text{Re}[\beta]+\max\limits_{\lambda_i\in\Lambda(A)}|\beta-\lambda_i|<0.
    \end{aligned}
\end{equation}
First, this requires $\text{Re}[\lambda_i]>0$ for all $\lambda_i\in \Lambda(A)$. Second, we must satisfy $\max\limits_{\lambda_i\in\Lambda(A)}\left|\text{Re}[\beta]-\text{Re}[\lambda_i]\right|<\text{Re}[\beta]$. This error analysis clearly demonstrates the necessity of the dissipative process (positive real eigenvalues) for the shift technique to succeed.
\par To show that dissipation is not sufficient, it suffices to provide a counterexample. Suppose we select $L(t)=A=\begin{bmatrix}0&1\\0&0\end{bmatrix}$. The matrix $\frac{A+A^\dagger}{2}$ is manifestly positive semidefinite. Utilizing the norm equality given by Eq.~(\ref{equ:dyson:8}), we have $\|\beta I-A\|=\frac{1+\sqrt{1+4|\beta|^2}}{2}\ge\text{Re}[\beta]$. Consequently, the shift technique fails.

\subsection{Further Improvements on the Shifted Dyson Series}
\label{section:8-3}

\begin{figure}[tbp]
    \centering
    \includegraphics[width=\linewidth]{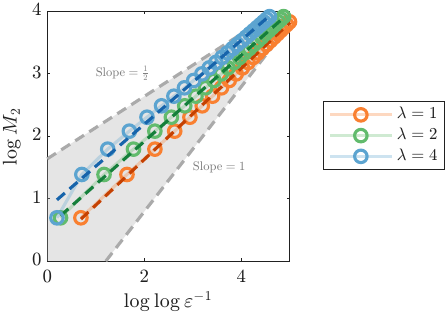}
    \caption{
    \label{figure:shifted-dyson-scaling}\textbf{Shifted Dyson truncation scaling.} For $\lambda=1,2,4$, the fitted slopes of $\log M_2$ versus $\log\log\varepsilon^{-1}$ lie between $\frac{1}{2}$ and $1$, consistent with the regime-dependent analysis.}
\end{figure}

\par This subSection~supplies the derivation behind the sharper regime-dependent estimate quoted in Section~\ref{section:3-1}, together with the numerical verification shown in Fig.~\ref{figure:shifted-dyson-scaling}. We notice that the complexity of constructing quantum circuits using the shifted Dyson series does not yet reach the theoretical bound in Eq.~(\ref{equ:2}). This discrepancy likely stems either from the inherent properties of the Dyson series or from our error bounding techniques. However, we found that the proof of Lemma~\ref{lemma:B1} can be optimized further when $M>2\tau$, hinting at potential improvements under specific conditions. Utilizing the result in Eq.~(\ref{equ:dyson:6}), we adopt a new analytical approach. Ensuring $e^{-\tau}\left(\frac{e\tau}{M}\right)^M < \varepsilon$ is equivalent to satisfying the following inequality:
\begin{equation*}
    \begin{aligned}
        -\tau+M\left(1+\log\left(\frac{\tau}{M}\right)\right) < \log\varepsilon.
    \end{aligned}
\end{equation*}
If $M \ge 2\tau$ holds, we can apply the inequality $\ln(1+x) \le x - \frac{1}{2}x^2$ for $x \in (-1,0]$. This simplifies the condition on $M$ to:
\begin{equation}
    \begin{aligned}
        \label{equ:dyson:15}
        \frac{(\tau-M)^2}{M}>2\log\varepsilon^{-1},
    \end{aligned}
\end{equation}
which under the condition $M>2\tau$ evaluates to:
\begin{equation*}
    \begin{aligned}
        M>\max\{2\tau,\tau+\log\varepsilon^{-1}+\sqrt{(\log\varepsilon^{-1})^2+2\tau\log\varepsilon^{-1}}\}.
    \end{aligned}
\end{equation*}
We observe here that as $\tau \to 0$, $M \sim \log \varepsilon^{-1}$ exhibits a linear dependence. Interestingly, as $\tau \to \infty$, $M \sim (\log \varepsilon^{-1})^{\frac{1}{2}}$ follows a square-root dependence. This suggests that the dependence of the general shifted Dyson series on $\log \varepsilon^{-1}$ theoretically lies in the interval $[\frac{1}{2}, 1]$, outperforming previous methods~\cite{An2023Optimal,Jin2025Optimal,Low2025Optimal}.

\par To verify this result, we perform a series expansion on the function $e^{-\lambda t}$ over the interval $t \in [0,1]$. We plot the scaling of the truncation term $\log M$ against the error scaling $\log \log \varepsilon^{-1}$. In Fig.~\ref{figure:shifted-dyson-scaling}, the slope of the fitted line lies precisely within $[\frac{1}{2}, 1]$. Furthermore, the slope decreases as $\lambda$ increases. This observation is perfectly consistent with our theoretical derivation.

\section{Proofs and Extensions of Phase-Driven APS}
\label{section:9}

\subsection{Non-Unitary Optimality = Phase-Driven APS + Shifted Dyson Series}
\label{section:9-1}
\par In this section, we consider the query complexity and gate complexity analysis for time-independent non-unitary dynamics. Based on the conclusion of Theorem~\ref{theorem:A1}, we need to address a time-independent unitary operator and a time-dependent Hermitian operator, where the former requires the application of the following classical quantum singular value transform (QSVT) algorithm~\cite{Gilyen2018singular}.
\begin{lemma}
\label{lemma:B3}
\textbf{(QSVT)~\cite{Gilyen2018singular}} 
Let $H$ be an $s$-sparse Hamiltonian acting on $m_H$ qubits. To simulate its evolution within error $\delta>0$, the query complexity scales as 
\begin{equation*}
    \begin{aligned}
        \mathcal{Q}_H=\mathcal{O}\left(\chi+\frac{\log\delta^{-1}}{\log\log\delta^{-1}}\right),
    \end{aligned}
\end{equation*}
and the gate complexity scales as
\begin{equation*}
    \begin{aligned}
        \mathcal{C}_H=\mathcal{O}\left(\chi\cdot\polylog(\chi,\delta^{-1})\right),
    \end{aligned}
\end{equation*}
where $\chi=H_{\max} t$ with $H_{\max}=\|H\|$, and $t$ denotes the total evolution time.
\end{lemma}
\par Based on this lemma, we obtain the following complexity analysis Theorem~for time-independent non-unitary dynamics.
\begin{theorem}
\label{theorem:B2}
\par For estimating the operator $e^{-At}$ with $A_1=\frac{A+A^\dagger}{2}\succeq 0$ over the time range $t\in[0,T]$, we can achieve a gate complexity that is logarithmic in $\varepsilon$, and the number of queries to $\mathrm{HAM}_P$ and $A$ is
\begin{equation}
    \begin{aligned}
        \label{equ:independent:1}
        \mathcal{O}\left(\frac{\|u_0\|}{\|u(t)\|}\left(A_{\max}t+\frac{\log\varepsilon^{-1}}{\log\log\varepsilon^{-1}}\right)\right),
    \end{aligned}
\end{equation}
and the number of queries to $u_0$ equals the number of repetitions, i.e.,
\begin{equation}
    \begin{aligned}
        \label{equ:independent:2}
        \mathcal{O}\left(\frac{\|u_0\|}{\|u(t)\|}\right).
    \end{aligned}
\end{equation}
Here $\mathrm{HAM}_P$ is defined from
\begin{equation}
    \begin{aligned}
        \label{equ:independent:3}
        (\bra{0}_a\otimes I_s)&\mathrm{HAM}_P(\ket{0}_a\otimes I_s)\\
        &=\sum\limits_{m=0}^{N_m-1}|m\rangle\langle m|\otimes\frac{e^{iA_2hm}\tilde{A}_1e^{-iA_2hm}}{A_{\max}},
    \end{aligned}
\end{equation}
where $\tilde{A}_1=A_1-A_{\max}I$ with $A_{\max}=\|A\|$, $\tau=tA_{\max}$, and the requirements for $M$, $N_m$ and $h$ satisfy the conditions stated in Theorem~\ref{theorem:B1}.
\begin{proof}
\par The proofs for the gate complexity and the per-run query complexity are provided in Lemma~\ref{lemma:B4}, while the calculation of the repetition count is presented in Lemma~\ref{lemma:B5}.
\end{proof}
\end{theorem}
\noindent This Theorem~is based on the following two lemmas, namely the per-run complexity analysis and the repetition count estimation.
\begin{lemma}
\label{lemma:B4}
\par For the following phase-driven APS of time-independent non-unitary dynamics
\begin{equation*}
    \begin{aligned}
        e^{-At}=e^{-iA_2t}\cdot e^{-\int_0^tA_p(s)\d s},
    \end{aligned}
\end{equation*}
where $A_p(t)=e^{iA_2t}A_1e^{-iA_2t}$ with $A_1=\frac{A+A^\dagger}{2}\succeq 0$. We can use the following Dyson series to approximate it with an error tolerance $\varepsilon$:
\begin{equation}
    \begin{aligned}
        \label{equ:independent:4}
        e^{-At}&\approx e^{-\tau}\sum\limits_{k=0}^M(-1)^kh^k\\
        &\quad\times\sum\limits_{0\le m_1<\cdots<m_k< N_m}e^{-iA_2t}\prod_{j=k}^1e^{iA_2m_jh}\tilde{A}_1e^{-iA_2m_jh},
    \end{aligned}
\end{equation}
where $\tilde{A}_1=A_1-A_{\max}I$ with $A_{\max}=\|A\|$, $\tau=tA_{\max}$, and the requirements for $M$, $N_m$ and $h$ satisfy the conditions stated in Theorem~\ref{theorem:B1}. Furthermore, for the specific quantum implementation, the query and gate complexity can be estimated as follows:\\
    $\bullet$ Qubits: $\mathcal{O}\left(\log\varepsilon^{-1}\right)$;\\
    $\bullet$ Preparation Queries to $A$: $\mathcal{O}\left(A_{\max}t+\frac{\log\varepsilon^{-1}}{\log\log\varepsilon^{-1}}\right)$;\\
    $\bullet$ Queries to $\mathrm{HAM}_P$: $\mathcal{O}\left(A_{\max}t+\frac{\log\varepsilon^{-1}}{\log\log\varepsilon^{-1}}\right)$;\\
    $\bullet$ Gates: $\mathcal{O}\left(A_{\max}t\cdot\polylog(A_{\max},t,\log \varepsilon^{-1},\varepsilon^{-1})\right)$,\\
where $j=1,\cdots,\lfloor\log N_m\rfloor$, and we use the fact $\|\tilde{A}_1\|,\|A_1\|,\|A_2\|\le A_{\max}$.
\begin{proof}
\par First, we provide an explanation for the construction of $\mathrm{HAM}_P$, which can be expressed in the following form~\cite{Low2019Interaction}:
\begin{equation}
    \begin{aligned}
        \label{equ:independent:5}
        \mathrm{HAM}_P&=\left(\sum\limits_{m=0}^{N_m-1}\ket{m}\bra{m}_d\otimes\mathbf{1}_a\otimes e^{iA_2hm}\right)\\
        &\quad\cdot(\mathbf{1}_d\otimes O_{A_1})\left(\sum\limits_{m=0}^{N_m-1}\ket{m}\bra{m}_d\otimes\mathbf{1}_a\otimes e^{-iA_2hm}\right),
    \end{aligned}
\end{equation}
where we assume the block-encoding of $A_1$ is fixed and easy to handle. For processing $e^{-iA_2hm}$ ($m=0,\cdots,N_m-1$), we can adopt a method similar to binary representation in computers, using the QSVT approach to implement $e^{-iA_2 h 2^j}$ ($j=1,\cdots,\lfloor\log N_m\rfloor$) quantumly~\cite{Low2019Interaction}, the query complexity to prepare $e^{-iA_2 h 2^j}$ is at most
\begin{equation}
    \begin{aligned}
    \mathcal{Q}=\mathcal{O}\left(A_{\max}t+\frac{\log\varepsilon^{-1}}{\log\log\varepsilon^{-1}}\right).
    \end{aligned}
\end{equation}
Since each $e^{-iA_2 h m}$ can be represented by at most $\lceil\log N_m\rceil$ existing operators with coefficient $1$, using the result  in Remark~\ref{remark:B3}, to achieve an error of $\varepsilon^\prime=\frac{\varepsilon}{\log N_m}$, the required overall gate complexity is at most
\begin{equation}
    \begin{aligned}
        \mathcal{C}&=\mathcal{O}\left(A_{\max}t\cdot\log N_m\cdot\polylog(A_{\max},t,\log N_m,\varepsilon^{-1})\right)\\
        &=\mathcal{O}\left(A_{\max}t\cdot\polylog(A_{\max},t,\log\varepsilon^{-1},\varepsilon^{-1})\right),
    \end{aligned}
\end{equation}
where we used the lower bound of $N_m$ from Theorem~\ref{theorem:B1}, i.e.
\begin{equation}
    \begin{aligned}
        \label{equ:independent:6}
        N_m\ge \frac{M(2A_{\max}^2t+A_{\max}(\tau+1))}{\varepsilon}=\mathcal{O}(MA_{\max}^2t\varepsilon^{-1}),
    \end{aligned}
\end{equation}
with $\hat{L}_{\max}=\|[A_1,A_2]\|\le2A_{\max}^2$. This implies that the gate complexity is logarithmic with respect to the error tolerance $\varepsilon$, which ensures the feasibility of the quantum algorithm. \par Furthermore, using the conclusions from Theorem~\ref{theorem:B1}, one can use a Dyson series to approximate $e^{-\int_0^tA_p(s)\d s}$ and obtain
\begin{equation}
    \begin{aligned}
        \label{equ:independent:7}
        e^{-At}&\approx e^{-\tau}\sum\limits_{k=0}^M(-1)^kh^k\\
        &\quad\sum\limits_{0\le m_1<\cdots<m_k< N_m}e^{-iA_2t}\prod_{j=k}^1e^{iA_2m_jh}\tilde{A}_1e^{-iA_2m_jh}.
    \end{aligned}
\end{equation}
For the number of qubits, using the calculation for $N_m$ in Eq.~(\ref{equ:independent:6}), we obtain
\begin{equation}
    \begin{aligned}
        \mathcal{B}=\mathcal{O}(\log N_m)=\mathcal{O}(\log\varepsilon^{-1}).
    \end{aligned}
\end{equation}
For the query and gate complexity for $e^{-\int_0^t A_p(s) \d s}$ regarding $\mathrm{HAM}_P$, one first needs to evaluate the norm of the coefficient for $Q_k$. It is straightforward to see that
\begin{equation}
    \begin{aligned}
        \label{equ:independent:8}
        \sum\limits_{k=0}^Me^{-\tau}h^k\sum\limits_{0\le m_1<\cdots<m_k< N_m}&A_{\max}^k\le \sum\limits_{k=0}^Me^{-\tau}\frac{\tau^k}{k!} = 1.
    \end{aligned}
\end{equation}
Applying the result from Theorem~\ref{theorem:B1} and Remark~\ref{remark:B3} and setting $\varepsilon^{\prime\prime}=\frac{\varepsilon}{M}$, we conclude that the number of required queries is at most $M$, and queries to $\mathrm{HAM}_P$ is at most
\begin{equation}
    \begin{aligned}
        \mathcal{Q}=\mathcal{O}\left(A_{\max}t+\frac{\log\varepsilon^{-1}}{\log\log\varepsilon^{-1}}\right).
    \end{aligned}
\end{equation}
This completes the proof.
\end{proof}
\end{lemma}

\begin{lemma}
\label{lemma:B5}
\par Using the oracle in Eq.~ (\ref{equ:independent:3}) to implement the quantum algorithm in Eq.~(\ref{equ:independent:4}) for time range $t\in[0,T]$, the repetition count scales as
\begin{equation}
    \begin{aligned}
        \label{equ:independent:9}
        g=\mathcal{O}\left(\frac{\|u_0\|}{\|u(t)\|}\right).
    \end{aligned}
\end{equation}
\begin{proof}
\par The success probability can be estimated as~\cite{An2023QuantumAF}
\begin{equation*}
    \begin{aligned}
        \textbf{Pr}&\sim\left(e^{-\tau}\sum\limits_{k=0}^Mh^k\sum\limits_{0\le m_1<\cdots<m_k< N_m}A_{\max}^k\right)^{-2}\|e^{-At}\|^2\\
        &\ge \left(\frac{\|u(t)\|}{\|u_0\|}\right)^2\left(\sum\limits_{k=0}^\infty e^{-\tau}\frac{\tau^k}{k!}\right)^{-2}\ge \left(\frac{\|u(t)\|}{\|u_0\|}\right)^2,
    \end{aligned}
\end{equation*}
where we used the inequality $\|u(t)\|\le \|e^{-At}\|\cdot\|u_0\|$. Using the amplitude amplification method, one can obtain an estimate of the number of repetitions as
\begin{equation}
    \begin{aligned}
        g=\mathcal{O}\left(\frac{\|u_0\|}{\|u(t)\|}\right).
    \end{aligned}
\end{equation}
This completes the proof.
\end{proof}
\end{lemma}
\subsection{Improved Query Complexity for Time-Dependent Case}
\label{section:9-2}
\par In this section, we further discuss the time-dependent non-unitary dynamics. For the analysis of the unitary operator, we rely on the truncated Dyson series proposed by Low et al.~\cite{Low2019Interaction}, with the specific content as follows:
\begin{lemma}
\label{lemma:B6}
\textbf{(Truncated Dyson Series)~\cite{Low2019Interaction}}
Let $H(s)$ be a time-dependent Hamiltonian, sufficiently smooth in $t$. To simulate its evolution within error $\delta>0$, the query complexity scales as
\begin{equation*}
    \begin{aligned}
        \mathcal{Q}_H=\mathcal{O}\left(\chi\frac{\log(\chi\delta^{-1})}{\log\log(\chi\delta^{-1})}\right),
    \end{aligned}
\end{equation*}
and the gate complexity scales as
\begin{equation*}
    \begin{aligned}
        \mathcal{C}_H=\mathcal{O}\left(\chi\cdot\polylog(\chi,\delta^{-1})\right),
    \end{aligned}
\end{equation*}
in which $\chi = H_{\max}t$ with $H_{\max}=\max\limits_{s\in[0,T]}\|H(s)\|_{\max}$, and $t$ denotes total evolution time.
\end{lemma}
\par Based on this lemma, we design a quantum algorithm for time-dependent non-unitary dynamics according to Theorem~\ref{theorem:A1}.
\begin{theorem}
\label{theorem:B3}
\par For estimating the operator $e^{-\int_0^tA(s)\d s}$ with $A_1(t)=\frac{A(t)+A^\dagger(t)}{2}\succeq 0$ over the time range $t\in[0,T]$, we can achieve a gate complexity that is logarithmic in $\varepsilon$, and the number of overall queries to $\mathrm{HAM}_{A_1}$ and $\mathrm{HAM}_{A_2}$ is
\begin{equation}
    \begin{aligned}
        \label{equ:dependent:1}
        \mathcal{O}\left(\frac{\|u_0\|}{\|u(t)\|}A_{\max}t\frac{\log\varepsilon^{-1}}{\log\log\varepsilon^{-1}}\right),
    \end{aligned}
\end{equation}
and the number of queries to $u_0$ equals the number of repetitions, i.e.,
\begin{equation}
    \begin{aligned}
        \label{equ:dependent:2}
        \mathcal{O}\left(\frac{\|u_0\|}{\|u(t)\|}\right).
    \end{aligned}
\end{equation}
The detailed definitions of $\mathrm{HAM}_{A_1}$ and $\mathrm{HAM}_{A_2}$ are provided below.
\begin{gather}
    \nonumber (\bra{0}_a\otimes I_s)\mathrm{HAM}_{A_1}(\ket{0}_a\otimes I_s)=\sum\limits_{m=0}^{N_m-1}\ket{m}\bra{m}\otimes\frac{A_1(hm)}{A_{\max}},\\
    \label{equ:dependent:3}(\bra{0}_a\otimes I_s)\mathrm{HAM}_{A_2}(\ket{0}_a\otimes I_s)=\sum\limits_{m=0}^{N_m-1}\ket{m}\bra{m}\otimes\frac{A_2(hm)}{A_{\max}},
\end{gather}
where $\tilde{A}_1=A_1-A_{\max}I$ with $A_{\max}=\|A\|$, $\tau=tA_{\max}$, $\mathcal{U}_p(t_i,t_j)=\mathcal{T}e^{-i\int_{t_i}^{t_j}H(s)\d s}$ with $t_{k+1}=t$, and the requirements for $M$, $N_m$ and $h$ satisfy the conditions stated in Theorem~\ref{theorem:B1}.
\begin{proof}
\par The proofs for the gate complexity and the per-run query complexity are provided in Lemma~\ref{lemma:B7}, while the calculation of the repetition count is presented in Lemma~\ref{lemma:B8}.
\end{proof}
\end{theorem}
Similarly, we require the following two lemmas, namely the per-run complexity analysis and the repetition count estimation.
\begin{lemma}
\label{lemma:B7}
\par For the following APS decomposition of time-dependent non-unitary dynamics
\begin{equation}
    \begin{aligned}
        \label{equ:dependent:4}
        \mathcal{T}e^{-\int_0^tA(s)\d s}=\mathcal{T}e^{-i\int_0^tA_2(s)\d s}\cdot\mathcal{T}e^{-\int_0^tA_p(s)\d s},
    \end{aligned}
\end{equation}
where $A_p(t)=\mathcal{U}_p^\dagger(t) A_1(t)\mathcal{U}_p(t)$ with $A_1(t)=\frac{A(t)+A^\dagger(t)}{2}\succeq 0$ and $\mathcal{U}_p(t)=\mathcal{T}e^{-i\int_0^t A_2(s)\d s}$. One can use the following Dyson series to approximate it with an error tolerance $\varepsilon$:
\begin{equation}
    \begin{aligned}
        \label{equ:dependent:5}
        &\quad\mathcal{T}e^{-\int_0^tA(s)\d s}\approx e^{-\tau}\sum\limits_{k=0}^M(-1)^kh^k\\
        \times&\hspace{-1.5em}\sum\limits_{0\le m_1<\cdots<m_k< N_m}\prod_{j=k}^1\left[\mathcal{U}_p(m_jh,m_{j+1}h)\cdot\tilde{L}(m_jh)\right]\cdot\mathcal{U}_p(m_1h),
    \end{aligned}
\end{equation}
where $\tilde{A}_1=A_1-A_{\max}I$ with $A_{\max}=\|A\|$, $\tau=tA_{\max}$, $\mathcal{U}_p(t_i,t_j)=\mathcal{T}e^{-i\int_{t_i}^{t_j}A_2(s)\d s}$ with $t_{k+1}=t$, and the requirements for $M$, $N_m$ and $h$ satisfy the conditions stated in Theorem~\ref{theorem:B1}. Furthermore, for the specific quantum implementation, the query and gate complexity can be estimated as follows:\\
    $\bullet$ Qubits: $\mathcal{O}\left(\log\varepsilon^{-1}\right)$;\\
    $\bullet$ Queries to $\mathrm{HAM}_{A_1}$: $\mathcal{O}\left(A_{\max}t+\frac{\log\varepsilon^{-1}}{\log\log\varepsilon^{-1}}\right)$;\\
    $\bullet$ Queries to $\mathrm{HAM}_{A_2}$: $\mathcal{O}\left(A_{\max}t\frac{\log\varepsilon^{-1}}{\log\log\varepsilon^{-1}}\right)$;\\
    $\bullet$ Gates: $\mathcal{O}\left((A_{\max}t)^2\cdot\polylog(A_{\max},t,\log\varepsilon^{-1},\varepsilon^{-1})\right)$,\\
where we use the fact $\|\tilde{A}_1\|,\|A_1\|,\|A_2\|\le A_{\max}$.
\begin{proof}
\par First, we provide the quantum implementation of $e^{-\int_0^t A(s) ds}$ by applying Theorem~\ref{theorem:B1}.
\begin{widetext}
\begin{equation}
    \begin{aligned}
        \mathcal{T}e^{-\int_0^tA(s)\d s}&\approx e^{-\tau}\sum\limits_{k=0}^M(-1)^kh^k\sum\limits_{0\le m_1<\cdots<m_k< N_m}\mathcal{U}_p(t)[A_p(m_kh)+A_{\max} I]\cdots[A_p(m_1h)+A_{\max} I]\\
        &=e^{-\tau}\sum\limits_{k=0}^M(-1)^kh^k\sum\limits_{0\le m_1<\cdots<m_k< N_m}\mathcal{U}_p(t)\prod_{j=k}^1\mathcal{U}_p^\dagger(m_jh)\tilde{L}(m_jh)\mathcal{U}_p(m_jh)\\
        &= e^{-\tau}\sum\limits_{k=0}^M(-1)^kh^k\sum\limits_{0\le m_1<\cdots<m_k< N_m}\prod_{j=k}^1\left[\mathcal{U}_p(m_jh,m_{j+1}h)\cdot\tilde{L}(m_jh)\right]\cdot\mathcal{U}_p(m_1h).
    \end{aligned}
\end{equation}
\end{widetext}
For the estimation for the number of qubits, one needs to calculate $\hat{L}_{\max}$:
\begin{equation*}
    \begin{aligned}
        \hat{L}_{\max}
        &\le \sup\limits_{s\in[0,T]}2\|A_1(s)\|\|A_2(s)\|+\left\|\left.\frac{\d L(t)}{\d t}\right|_{t=s}\right\|\\
        &\le 2A_{\max}^2+\sup\limits_{s\in[0,T]}\left\|\left.\frac{\d A_1(t)}{\d t}\right|_{t=s}\right\|,
    \end{aligned}
\end{equation*}
and the number of discrete intervals $N_m$ should be chosen as
\begin{equation}
    \begin{aligned}
        N_m&\ge \frac{M(2A_{\max}^2t+\sup\limits_{s\in[0,T]}\left\|\left.\frac{\d A_1(t)}{\d t}\right|_{t=s}\right\|+A_{\max}(\tau+1))}{\varepsilon}\\
        &=\mathcal{O}(\varepsilon^{-1}),
    \end{aligned}.
\end{equation}
Then, we can obtain
\begin{equation}
    \begin{aligned}
        \mathcal{B}=\mathcal{O}(\log N_m)=\mathcal{O}(\log\varepsilon^{-1}).
    \end{aligned}
\end{equation}
For query complexity analysis, based on the conclusions from Eq.~(\ref{equ:independent:8}), Lemma~\ref{lemma:B6}, and Remark~\ref{remark:B3}, we can set $\varepsilon^\prime = \frac{\varepsilon}{M}$ and find that the number of queries to $\mathrm{HAM}_{L}$ is
\begin{equation}
    \begin{aligned}
        \mathcal{Q}=\mathcal{O}(M)=\mathcal{O}\left(A_{\max}t+\frac{\log\varepsilon^{-1}}{\log\log\varepsilon^{-1}}\right).
    \end{aligned}
\end{equation}
However, for $\mathcal{U}_p(t_j,t_{j+1})$ and $\mathcal{U}_p(t_j)$, one cannot apply the method from Theorem~\ref{theorem:B2} since $A(t)$ is time-dependent. In this case, one can only directly query the oracle for $A$ and set $\varepsilon^{\prime\prime}=\frac{\varepsilon}{M}$, resulting in a query complexity to $\mathrm{HAM}_{H}$:
\begin{equation}
    \begin{aligned}
        \mathcal{Q}&=\sum\limits_{j=1}^k\mathcal{O}\left(A_{\max}(m_{j+1}-m_j)h\frac{\log\varepsilon^{-1}}{\log\log\varepsilon^{-1}}\right)\\
        &=\mathcal{O}\left(A_{\max}t\frac{\log\varepsilon^{-1}}{\log\log\varepsilon^{-1}}\right),
    \end{aligned}
\end{equation}
for any $k\in1,\cdots,M$. Finally, for the gate complexity analysis, since we need to construct $M$ unitary operators, the gate complexity is
\begin{equation}
    \begin{aligned}
        \mathcal{C}&=\mathcal{O}\left(A_{\max}t\cdot M\cdot\polylog(A_{\max},t,\log N_m,\varepsilon^{-1})\right)\\
        &=\mathcal{O}\left((A_{\max}t)^2\cdot\polylog(A_{\max},t,\log\varepsilon^{-1},\varepsilon^{-1})\right).
    \end{aligned}
\end{equation}
This completes the proof.
\end{proof}
\end{lemma}
\begin{lemma}
\label{lemma:B8}
\par Using the oracle in Eq.~(\ref{equ:dependent:3}) to implement the quantum algorithm shown in Eq.~(\ref{equ:dependent:5}) over the time interval $t\in[0,T]$, the repetition count scales as
\begin{equation}
    \begin{aligned}
        g=\mathcal{O}\left(\frac{\|u_0\|}{\|u(t)\|}\right).
    \end{aligned}
\end{equation}
\begin{proof}
\par The success probability can be estimated as~\cite{An2023QuantumAF}
\begin{equation*}
    \begin{aligned}
        \textbf{Pr}&\sim\left(e^{-\tau}\sum\limits_{k=0}^Mh^k\sum\limits_{0\le m_1<\cdots<m_k< N_m}A_{\max}^k\right)^{-2}\|\mathcal{T}e^{-\int_0^tA(s)\d s}\|^2\\
        &\ge \left(\frac{\|u(t)\|}{\|u_0\|}\right)^2\left(\sum\limits_{k=0}^\infty e^{-\tau}\frac{\tau^k}{k!}\right)^{-2}\ge \left(\frac{\|u(t)\|}{\|u_0\|}\right)^2,
    \end{aligned}
\end{equation*}
where we use the inequality $\|u(t)\|\le \|e^{-\int_0^tL(s)\d s}\|\cdot\|u_0\|$. Using the amplitude amplification method, we can obtain an estimate of the number of repetitions as
\begin{equation*}
    \begin{aligned}
        g=\mathcal{O}\left(\frac{\|u_0\|}{\|u(t)\|}\right).
    \end{aligned}
\end{equation*}
\end{proof}
\end{lemma}

\subsection{Ill-Posed Problems}
\label{section:9-3}
\par Throughout the proof of Theorem~\ref{theorem:main:1}, we assumed $A_1=\frac{A+A^\dagger}{2}\succeq 0$, because Theorem~\ref{theorem:B1} strictly requires positive semi-definiteness. However, in practice, this condition often fails for problems with internal instabilities~\cite{Jin2024IllPosed}. Let $\lambda_{\min}^-(A)=\sup\limits_{s\in [0,t]}\lambda_{\min}(A_1(s))<0$. Next, we will relax the conditions of Theorem~\ref{theorem:B1} by examining cases where $\lambda_{\min}^-(A)t=\mathcal{O}(1)$. We need to readjust the parameters accordingly. 
\par Let the shifted matrix be $\tilde{A}_1(t)=A_1(t)-\frac{A_{\max}+\lambda_{\min}^-(A)}{2} I$, and let $\tau=\frac{A_{\max}+\lambda_{\min}^-(A)}{2}t$. We then bound the shifted matrix norm by $\tilde{L}_{\min}=\sup\limits_{s\in[0,t]}\|\tilde{A}_1(s)\|\le \frac{A_{\max}-\lambda_{\min}^-(A)}{2}$.
\par The necessary modifications target two primary factors. The first is the polynomial degree derived in Theorem~\ref{theorem:B1}, which dictates the per-run query complexity computed in Theorem~\ref{theorem:B2}. The second is the number of repetitions evaluated in Lemma~\ref{lemma:B5}. Both of these factors hinge completely on the following 1-norm estimation:
\begin{equation}
    \begin{aligned}
        \label{equ:extend:1}
        \sum\limits_{k=1}^M\|Q_k\|&\le e^{-\tau}h^k\sum\limits_{0\le m_1<\cdots<m_k< N_m}\prod_{j=1}^k\|\tilde{A}_1(m_kh)\|\\
        &\le e^{-\tau}\sum\limits_{k=1}^M\frac{(\frac{A_{\max}-\lambda_{\min}^-(A)}{2})^k}{k!}\le e^{\lambda_{\min}^-(A)t}.
    \end{aligned}
\end{equation}
Compared to the derivation in Lemma~\ref{lemma:B1}, we clearly observe that the sole difference stems from the multiplicative coefficient $e^{\lambda_{\min}^-(A)t}$. Therefore, one simply sets $\varepsilon^\prime = \varepsilon e^{-\lambda_{\min}^-(A)t}$ and substitutes it back into Lemma~\ref{lemma:B1}'s conclusion. This directly yields the requisite polynomial order $M$ as:
\begin{equation}
    \begin{aligned}
        \label{equ:extend:2}
        M&=\mathcal{O}\left(\tau+\lambda_{\min}^-(A)t\frac{\log\varepsilon^{-1}}{\log\log\varepsilon^{-1}}\right)\\
        &=\mathcal{O}\left(\frac{A_{\max}+\lambda_{\min}^-(A)}{2}T+\lambda_{\min}^-(A)t\frac{\log\varepsilon^{-1}}{\log\log\varepsilon^{-1}}\right)\\
        &=\mathcal{O}\left(A_{\max}t+\frac{\log\varepsilon^{-1}}{\log\log\varepsilon^{-1}}\right).
    \end{aligned}
\end{equation}
Regarding the number of repetitions, its value scales directly with the norm magnitude established in Eq.~(\ref{equ:extend:1}). Under the assumption that $\lambda_{\min}^-(A)t=\mathcal{O}(1)$, one can successfully apply amplitude amplification to bound the repetition count at:
\begin{equation}
    \begin{aligned}
        \label{equ:extend:3}
        g=\mathcal{O}\left(e^{\lambda_{\min}^-(A)t}\frac{\|u_0\|}{\|u(t)\|}\right).
    \end{aligned}
\end{equation}
Combining Eqs.~(\ref{equ:extend:2}) and (\ref{equ:extend:3}), we deduce that the total query complexity under the limit $\lambda_{\min}^-(A)t=\mathcal{O}(1)$ becomes:
\begin{equation}
    \begin{aligned}
        \label{equ:extend:4}
        \mathcal{Q}&=\mathcal{O}\left(e^{\lambda_{\min}^-(A)t}\frac{\|u_0\|}{\|u(t)\|}\left(A_{\max}t+\frac{\log\varepsilon^{-1}}{\log\log\varepsilon^{-1}}\right)\right)\\
        &=\mathcal{O}\left(\frac{\|u_0\|}{\|u(t)\|}\left(A_{\max}t+\frac{\log\varepsilon^{-1}}{\log\log\varepsilon^{-1}}\right)\right).
    \end{aligned}
\end{equation}
This outcome strongly indicates that our conclusion regarding non-unitary optimal scaling holds even for mildly ill-posed problems.

\subsection{Inhomogeneous Term}
\label{section:9-4}
\par For systems involving inhomogeneous terms, we employ the dimension expansion method proposed by Jin et al.~\cite{Jin2025LinearSystems}. By introducing a new solution variable $\tilde{u}(t)=[u(t); K\mathbf{1}]$, we express the expanded ODE and the new initial condition as:
\begin{equation}
    \begin{aligned}
        \label{equ:extend:5}
        \frac{\d\tilde{u}(t)}{\d t}=\begin{bmatrix}
            A(t) & \frac{\text{diag}[b(t)]}{K}\\
            O & 0
        \end{bmatrix}\tilde{u}(t):=\tilde{A}(t)\tilde{u}(t),\quad \tilde{u}(0)=\begin{bmatrix}
            u_0\\
            K\mathbf{1}
        
        \end{bmatrix},
    \end{aligned}
\end{equation}
where $K$ is a free parameter. To solve this updated ODE in Eq.~(\ref{equ:extend:5}), we treat it as an expanded-state case that can be slightly ill-posed. Provided that $\lambda_{\min}^-(\tilde{A}_1)t=\mathcal{O}(1)$, we can implement an optimal-query quantum algorithm precisely as outlined earlier. Hence, the goal is to choose a parameter $K$ that satisfies this condition. As a sufficient bound, we require that the following constraint uniformly holds:
\begin{equation*}
    \begin{aligned}
        \lambda_{\min}(\tilde{A}_1)t\ge \lambda_{\min}(A_1)t-\frac{b_{\max}t}{NK}=\mathcal{O}(1),
    \end{aligned}
\end{equation*}
where $b_{\max}=N\sup\limits_{s\in[0,t]}\max\limits_{i=1}^N|b_i(t)|$. Effectively, if we simply select $K=\frac{b_{\max}t}{N}$, we can successfully bridge the gap and directly apply the formal conclusions developed for ill-posed problems.

\section{Additional Discussion of Hermitian Fast-Forwarding}
\label{section:10}

\subsection{Piecewise Time-Independent Case}
\label{section:10-1}
\par For general time-dependent Hermitian operators, does Eq.~(\ref{equ:fast:2}) still hold? Unfortunately, directly generalizing this conclusion is difficult. However, we note that Eq.~(\ref{equ:fast:2}) exhibits a specific regularity. Treating $\eta$ as a random variable $\xi \sim \mathcal{N}(0,2)$ with mean $\mu=0$ and variance $\sigma^2=2$, the time-independent expectation of $\mathcal{T}e^{i\xi\sqrt{Lt}}$ evaluates to
\begin{equation*}
    \begin{aligned}
        \mathbb{E}_{\eta}\left[e^{-i\eta\sqrt{Lt}}\right]=\frac{1}{2\sqrt{\pi}}\int_{-\infty}^{+\infty}e^{-\frac{\eta^2}{4}}e^{-i\eta\sqrt{Lt}}\d\eta=e^{-Lt}.
    \end{aligned}
\end{equation*}
Remarkably, this perfectly matches Eq.~(\ref{equ:fast:2}). From a probabilistic perspective, this exactly recovers the time-independent result. We expect the same idea to extend to LCHS, although we leave a detailed treatment for future work.
\par We extend this result to the piecewise time-independent case. Partition the time-dependent matrix $L(t)$ on $[0,t]$ into $N_t$ segments $[s_jt,s_{j+1}t]$ ($j=0,\cdots,N_t-1$). Let $L(st)=L_j$ for $s\in[s_j,s_{j+1}]$, where $L_j$ denotes the time-independent matrix on the $j$th segment. We handle this using stochastic methods. Let $\eta(t)$ be piecewise constant, taking independent Gaussian values $\xi_j$ on each segment, with $\xi_j\sim\mathcal{N}\!\left(0,\frac{2}{s_{j+1}-s_j}\right)$. We derive the following properties:

\begin{equation}
    \begin{aligned}
        \label{equ:fast:4}
        E[\eta(t)]=0,\quad E[\eta(t_1)\eta(t_2)]=\left\{\begin{array}{l}
             \frac{2}{s_{j+1}-s_j},t_1,t_2\in [s_j,s_{j+1}],\\
             0,\text{else.}
        \end{array}\right.
    \end{aligned}
\end{equation}
Hence, we can decompose the operator $\mathcal{T}e^{-i\int_0^1\eta(s^\prime)\sqrt{tL(s^\prime t)}\,\mathrm{d}s^\prime}$ into a product formula based on the time segments. Using the time-independent result from Eq.~(\ref{equ:fast:2}), we obtain the conclusion for the piecewise time-independent case:
\begin{equation}
    \begin{aligned}
        \label{equ:fast:9}
        &\mathcal{T}e^{-\int_0^tL(s)\d s}=\mathcal{T}e^{-\int_0^1tL(ts^\prime)\d s^\prime}=\prod\limits_{j=0}^{N_t-1}e^{-(s_{j+1}-s_j)tL_j}\\
        =&\prod\limits_{j=0}^{N_t-1}E_\eta\left[e^{-i\xi_j(s_{j+1}-s_j)\sqrt{tL_j}}\right]
        =E_\eta\left[\mathcal{T}e^{-i\int_0^1\eta(s^\prime)\sqrt{tL(s^\prime t)}\d s^\prime}\right].
    \end{aligned}
\end{equation}
We derived a representation for a special time-dependent case through a stochastic construction. The same approach may also extend to other quantum algorithmic constructions.
\par For the Hamiltonian simulation in Eq.~(\ref{equ:fast:9}), the variable $\xi_j$ follows an unbounded Gaussian distribution. We adopt the exact truncation technique from Eq.~(\ref{equ:fast:3}) by setting a threshold $M$. Define the truncated random variable $\overline{\xi}_j$ by setting $\overline{\xi}_j=\xi_j$ when $|\xi_j|\le M$ and $\overline{\xi}_j=0$ otherwise. We must pick $M$ such that the truncation error remains strictly within the tolerance $\varepsilon$. Applying the piecewise method of Eq.~(\ref{equ:fast:9}), we set $\Delta s = \min\limits_{j=0}^{N_t-1}(s_{j+1}-s_j) = \mathcal{O}(1)$. This yields:
\begin{widetext}
\begin{equation}
    \begin{aligned}
        \label{equ:fast:10}
        \left\|E_\eta\left[\mathcal{T}e^{-i\int_0^1\eta(s^\prime)\sqrt{tL(s^\prime t)}\d s^\prime}-\mathcal{T}e^{-i\int_0^1\overline{\eta}(s^\prime)\sqrt{tL(s^\prime t)}\d s^\prime}\right]\right\|
        \le& N_t\left\|E_{\xi_j}\left[e^{-i\xi_j(s_{j+1}-s_j)\sqrt{tL_j}}-e^{-i\overline{\xi}_j(s_{j+1}-s_j)\sqrt{tL_j}}\right]\right\|\\
        \le&2N_t\int_M^{+\infty}e^{-\frac{\xi_j^2(s_{j+1}-s_j)}{4}}|s_{j+1}-s_j|\sqrt{tL_{\max}}\d\xi_j\\
        \le&2N_t\int_M^{+\infty}e^{-\frac{\xi_j^2(s_{j+1}-s_j)}{4}}\sqrt{tL_{\max}}\d\xi_j,
    \end{aligned}
\end{equation}
\end{widetext}
We restricted the analysis to the range $|\xi_j|>M$. We also applied the inequality $\left\|\prod\limits_{j=1}^qA_j-\prod\limits_{j=1}^qB_j\right\|\le\sum\limits_{k=1}^q\left(\prod\limits_{j=1}^{k-1}\|A_j\|\right)\|A_k-B_k\|\left(\prod\limits_{j=k+1}^q\|B_j\|\right)$ mapped with the bounds $\max\limits_{j=0}^{N_t-1}(s_{j+1}-s_j) \le 1$. For $M>1$, we use a Gaussian tail bound to derive an upper limit for Eq.~(\ref{equ:fast:10}):

\begin{equation*}
    \begin{aligned}
        2N_t\int_M^{+\infty}&e^{-\frac{\xi_j^2(s_{j+1}-s_j)}{4}}\sqrt{tL_{\max}}\d\xi_j\\
        &\le \frac{4N_t\sqrt{t L_{\max}}}{s_{j+1}-s_j}e^{-\frac{(s_{j+1}-s_j) M^2}{4}}\le \varepsilon.
    \end{aligned}
\end{equation*}
Assuming $\Delta s = \mathcal{O}(1)$ and $N_t=\mathcal{O}(1)$, we can select $M$ fulfilling the required condition to guarantee precise error bounds:
\begin{equation}
    \begin{aligned}
        \label{equ:fast:11}
        M=\mathcal{O}\left(\sqrt{\log\varepsilon^{-1}}\right).
    \end{aligned}
\end{equation}
This effectively limits the independent query complexity per run to $\sqrt{L_{\max}t\log\varepsilon^{-1}}$.
\par Furthermore, the success probability of simulating the unitary operator $\mathcal{T}e^{-i\int_0^1\eta(s^\prime)\sqrt{tL(s^\prime t)}\d s^\prime}$ is roughly $1$, and the averaging weight is normalized. The required number of repetitions is thus $\mathcal{O}\big(\|u_0\|/\|u(t)\|\big)$. Bringing this all together, the overall query complexity for the piecewise time-independent case becomes:
\begin{equation}
    \begin{aligned}
        \label{equ:fast:12}
        \mathcal{Q}=\mathcal{O}\left(\frac{\|u_0\|}{\|u(t)\|}\sqrt{L_{\max}t\log\varepsilon^{-1}}\right).
    \end{aligned}
\end{equation}

\subsection{Non-Unitary Fast-Forwarding = Amplitude-Driven APS + Dyson Series}
\label{section:10-2}
\par Although existing research has developed truncated Dyson-series methods for time-dependent unitary operators, in the present $\mathcal{U}_a^{-1}$-based amplitude-driven APS the factor $\mathcal{T}e^{-i\int_0^tA_a(s)ds}$ is generally non-unitary. Therefore, those methods are not directly applicable to the operator $e^{-A_1t}\cdot\mathcal{T}e^{-i\int_0^tA_a(s)ds}$ without adaptation. Here, we provide an approximation and a quantum implementation tailored to this operator.
\begin{theorem}
\label{theorem:C1}
\par Let $A_1$ be a Hermitian matrix satisfying $A_1\succeq 0$, and define $A_{1,\max}=\|A_1\|$, $A_{2,\max}=\|A_2\|$, $\tau_2=tA_{2,\max}$. For any error tolerance $0<\varepsilon<2^{-e}$, there exists a quantum algorithm for simulating the evolution in Eq.~(\ref{equ:interaction:1}) whose query complexity achieves the fast-forwarding bound with respect to both matrix norm and evolution time:
\begin{equation}
    \begin{aligned}
        \label{equ:fast:13}
        \mathcal{Q}=\mathcal{O}\left(\frac{\|u_0\|}{\|u(t)\|}\left(\sqrt{A_{1,\max}T\log\varepsilon^{-1}}+A_{2,\max}T\frac{\log\varepsilon^{-1}}{\log\log\varepsilon^{-1}}\right)\right),
    \end{aligned}
\end{equation}
and the number of queries to $u_0$ equals the number of repetitions, i.e.,
\begin{equation}
    \begin{aligned}
        \label{equ:fast:14}
        \mathcal{O}\left(\frac{\|u_0\|}{\|u(t)\|}\right).
    \end{aligned}
\end{equation}
\begin{proof}
\par First, the number of oracles required to construct the quantum circuit is given in Lemma~\ref{lemma:C1}. For the success probability, the part involving $e^{-A_1t}$ requires $\mathcal{O}\left(\frac{\|u_0\|}{\|u(t)\|}\right)$ repetitions, as shown in Section~\ref{section:4-1}. Meanwhile, the factor $\mathcal{T}e^{-i\int_0^tA_a(s)\d s}$ is treated as a general operator and is implemented within the per-run query bound in Lemma~\ref{lemma:C1}; in the stated model, it does not introduce an additional asymptotic repetition factor. Thus, the total number of repetitions remains $\mathcal{O}\left(\frac{\|u_0\|}{\|u(t)\|}\right)$. This completes the proof.
\end{proof}
\end{theorem}

\begin{lemma}
\label{lemma:C1}
\par Let $A_1$ be a Hermitian matrix satisfying $A_1\succeq 0$, and define $A_{1,\max}=\|A_1\|$, $A_{2,\max}=\|A_2\|$, $\tau_2=tA_{2,\max}$. For any error tolerance $0<\varepsilon<2^{-e}$, decompose $[0,t]$ into $N_t=\mathcal{O}(\tau_2)$ equidistant subintervals $[t_j,t_{j+1}]$ ($j=1,\cdots,N_t-1$) with $t_0=0$ and $t_{N_t}=t$, each with an equal length smaller than $\log 2$. Then we insert $N_m$
points into each time subinterval $[t_j,t_{j+1}]$ and define $h=\frac{t}{N_m}$. If the truncation order $M$ and the number of discrete intervals $N_m$ for each time interval are chosen such that
\begin{equation}
    \begin{aligned}
        \label{equ:fast:15}
        M&=\mathcal{O}\left(\frac{\log\varepsilon^{-1}}{\log\log\varepsilon^{-1}}\right),\\ N_m&\ge \frac{Me^{\tau_2}(A_{1,\max}A_{2,\max}t+A_{2,\max}(\tau_2+1))}{\varepsilon},
    \end{aligned}
\end{equation}
then the operator $e^{-At}$ can be approximated by the following series with error at most $\varepsilon$:
\begin{equation*}
    \begin{aligned}
        \left\|e^{-At}-\left[\sum\limits_{k=0}^M(-i)^kY_k\right]^{N_t}\right\|\le\varepsilon,
    \end{aligned}
\end{equation*}
where $Y_k$ is defined as
\begin{equation*}
    \begin{aligned}
        Y_k=h^k\sum\limits_{0\le m_1<\cdots<m_k< N_m}\left(\overleftarrow{\prod_{j=1}^k}e^{-A_1(m_{j+1}-m_j)h}A_2\right)e^{-A_1m_1h},
    \end{aligned}
\end{equation*}
in which $m_1,\cdots,m_k$ are indices in $0,\cdots,N_m$ and $m_{k+1}=N_m$. Furthermore, for the specific quantum implementation, the query and gate complexity can be estimated as follows:\\
    $\bullet$ Preparation Queries to $A_1$: $\mathcal{O}\left(\sqrt{A_{1,\max}t\log\varepsilon^{-1}}\right)$;\\
    $\bullet$ Queries to $\mathrm{HAM}_A$: $\mathcal{O}\left(A_{2,\max}t\frac{\log\varepsilon^{-1}}{\log\log\varepsilon^{-1}}\right)$;\\
    $\bullet$ Queries to $A_2$: $\mathcal{O}\left(A_{2,\max}t\frac{\log\varepsilon^{-1}}{\log\log\varepsilon^{-1}}\right)$,\\
where the oracle $\mathrm{HAM}_A$ is defined as
\begin{equation}
    \begin{aligned}
        \label{equ:fast:16}
        (\bra{0}_a\otimes I_s)\mathrm{HAM}_A&(\ket{0}_a\otimes I_s)
        =\sum\limits_{m=0}^{N_m-1}|m\rangle\langle m|\otimes e^{-\frac{A_1hm}{N_t}}.
    \end{aligned}
\end{equation}
\begin{proof}
\par As indicated in Lemma~\ref{lemma:C2}, we need to set $\tau_2<\log 2$. Therefore, we need to add $N_t$ points to decompose the time interval $[0,t]$ into segments of length less than $\log 2$. Thus, we have $e^{-At}=\left[e^{-\frac{At}{N_t}}\right]^{N_t}$, and the error estimation is as follows:
\begin{equation}
    \begin{aligned}
        &\left\|\left[e^{-\frac{At}{N_t}}\right]^{N_t}-\left[\sum\limits_{k=0}^M(-i)^kY_k\right]^{N_t}\right\|\\
        \le& \sum\limits_{k=0}^{N_t-1}\left\|e^{-\frac{At}{N_t}}\right\|^k\left\|e^{-\frac{At}{N_t}}-\sum\limits_{k=0}^M(-i)^kY_k\right\|\left\|\sum\limits_{k=0}^M(-i)^kY_k\right\|^{N_t-k-1}\\
        \le& N_t\left\|e^{-\frac{At}{N_t}}-\sum\limits_{k=0}^M(-i)^kY_k\right\|,
    \end{aligned}
\end{equation}
where we use the fact $\left\|\sum\limits_{k=0}^M(-i)^kY_k\right\|\le 1$. To apply the conclusions from Lemmas \ref{lemma:C2} and \ref{lemma:C3}, we need to set $\varepsilon^\prime = \frac{1}{N_t}\varepsilon$ and $t^\prime=\frac{t}{N_t}$, and derive the requirements for $M$ and $N_t$ in each segment as:
\begin{equation}
    \begin{aligned}
        \label{equ:fast:17}
        M&=\mathcal{O}\left(\frac{\log\varepsilon^{-1}}{\log\log\varepsilon^{-1}}\right),\\
        N_m&\ge \frac{Me^{\tau_2}(A_{1,\max}A_{2,\max}t+A_{2,\max}(\tau_2+1))}{\varepsilon}.
    \end{aligned}
\end{equation}
\par Furthermore, we consider the overall query complexity of the approximation scheme. Since the term $A_p(t)$ in $Y_k$ involves the construction of $e^{-A_1 t}$, we decompose it into the following form:
\begin{equation*}
    \begin{aligned}
        Y_k&=h^k\sum\limits_{0\le m_1<\cdots<m_k< N_m}e^{-\frac{A_1(t-m_kh)}{N_t}}A_2e^{-\frac{A_1(m_k-m_{k-1})h}{N_t}}\\
        &\quad\times\cdots e^{-\frac{A_1(m_2-m_1)h}{N_t}}A_2e^{-\frac{m_1h}{N_t}}.
    \end{aligned}
\end{equation*}
It can be seen that the key lies in the construction of $e^{-\frac{A_1hm}{N_t}}$ (where $m=1,\cdots, N_m$). Therefore, the $\mathrm{HAM}_A$ we construct addresses this issue. Using the same setup as in Section~\ref{section:3}, we only need to design circuits for $e^{-\frac{A_1h2^j}{N_t}}$ (where $j=1,\cdots,\lceil\log N_m\rceil$). The remaining operators can be implemented with a gate complexity of at most $\mathcal{O}(\log N_m)$. Thus, the query complexity for preparing $\mathrm{HAM}_A$ can be obtained from the conclusions in Section~\ref{section:4-1}, i.e., $\mathcal{O}\left(\sqrt{A_{1,\max}t\log\varepsilon^{-1}}\right)$. On the other hand, the number of oracle calls for $\mathrm{HAM}_A$ and $A_2$ is determined by $N_tM$, which is $\mathcal{O}\left(A_{2,\max} t \frac{\log \varepsilon^{-1}}{\log \log \varepsilon^{-1}}\right)$. This completes the proof.
\end{proof}
\end{lemma}
\begin{lemma}
\label{lemma:C2}
\par Let $A_1$ be a Hermitian matrix satisfying $A_1\succeq 0$, and define $A_{1,\max}=\|A_1\|$, $A_{2,\max}=\|A_2\|$, $\tau_2=tA_{2,\max}<\log 2$. For any error tolerance $0<\varepsilon<2^{-e}$, if the truncation order $M$ is chosen such that 
\begin{equation}
    \begin{aligned}
        \label{equ:unitary:1}
        M=\mathcal{O}\left(\frac{\log\varepsilon^{-1}}{\log\log\varepsilon^{-1}}\right),
    \end{aligned}
\end{equation}
with the implied multiplicative constant at most $e^2$, then the amplitude-driven APS $e^{-A_1t}\cdot\mathcal{T}e^{-i\int_0^t A_a(s)\d s}$ can be approximated by a $M$-term series with error at most $\varepsilon$:
\begin{equation*}
    \begin{aligned}
        \left\|e^{-A_1t}\cdot\mathcal{T}e^{-i\int_0^tA_a(s)\d s}-\sum\limits_{k=0}^M(-i)^kX_k\right\|\le\varepsilon,
    \end{aligned}
\end{equation*}
where $X_k$ is the $k-$th term in the Dyson series, given by:
\begin{equation}
    \begin{aligned}
        \label{equ:unitary:2}
        X_k=\frac{1}{k!}\int_0^t\cdots\int_0^te^{-A_1t}\mathcal{T}\left[A_a(t_1)\cdots A_a(t_k)\right]\d^k t.
    \end{aligned}
\end{equation}
\begin{proof}
\par It should be noted that the definition of $X_k$ in Eq.~(\ref{equ:unitary:2}) admits the following equivalent representation:
\begin{equation*}
    \begin{aligned}
        X_k&=e^{-A_1t}\int_0^t\d t_kA_a(t_k)\cdots\int_0^{t_2}dt_1A_a(t_1)\\
        =&\int_0^t\cdots \int_0^{t_2}e^{-A_1(t-t_k)}A_2e^{-A_1(t_k-t_{k-1})}\\
        &\qquad\times\cdots e^{-A_1(t_2-t_1)}A_2e^{-A_1t_1}\d t_1\cdots \d t_k,
    \end{aligned}
\end{equation*}
Based on this equivalent formula and the fact that $A_1$ is a positive semidefinite matrix, we can derive the following upper bound estimate for $\|X_k\|$:
\begin{equation*}
    \begin{aligned}
        \|X_k\|&\le\int_0^t\cdots\int_0^{t_2}\prod_{j=1}^k\left\|A_2\right\|\d t_1\cdots \d t_k\le \frac{\tau_2^k}{k!}.
    \end{aligned}
\end{equation*}
Therefore, under the condition $M>2\tau_2$, we obtain the following bound:
\begin{equation}
    \begin{aligned}
        \label{equ:unitary:3}
        &\left\|e^{-A_1t}\cdot\mathcal{T}e^{-i\int_0^tA_a(s)\d s}\right.-\left.\sum\limits_{k=0}^M(-i)^kX_k\right\|\\
        \le&\sum\limits_{k=M+1}^\infty\left\|X_k\right\|
        \le \sum\limits_{k=M}^\infty\frac{\tau_2^k}{k!}\le \frac{\tau_2^k}{M!}\sum\limits_{k=M}^\infty\frac{1}{2^{k-M}}\le \left(\frac{e\tau_2}{M}\right)^M,
    \end{aligned}
\end{equation}
where this equation can also be viewed as a transformation involving the Lambert function $\mathcal{W}(x)$, and we can derive a lower bound for $M$:
\begin{equation}
    \begin{aligned}
        \label{equ:unitary:4}
        M\ge\frac{\log\varepsilon^{-1}}{\mathcal{W}(\frac{\log\varepsilon^{-1}}{e\tau})},
    \end{aligned}
\end{equation}
where we have set $\tau_2\le\min\{\log\varepsilon^{-1},e\log 2\}$. Using the inequality $\mathcal{W}(x) \ge \frac{\log x + 1}{2}$ for $x \ge 1$~\cite{Low2017Optimal}, we can establish a lower bound for $M$ via the third term in Eqs. (\ref{equ:unitary:4}):
    \begin{equation}
        \begin{aligned}
            \label{equ:unitary:5}
            M\ge \left\{\frac{2\log\varepsilon^{-1}}{1+\log\log\varepsilon^{-1}},2\tau_2\right\}=\mathcal{O}\left(\frac{\log\varepsilon^{-1}}{\log\log\varepsilon^{-1}}\right).
        \end{aligned}
    \end{equation}
This completes the proof.
\end{proof}
\end{lemma}
\begin{lemma}
\label{lemma:C3}
\par Let $A_1$ be a Hermitian matrix satisfying $A_1\succeq 0$, and define $A_{1,\max}=\|A_1\|$, $A_{2,\max}=\|A_2\|$, $\tau_2=tA_{2,\max}<\log 2$. For any error tolerance $0<\varepsilon<2^{-e}$, define $m_k = \left\lfloor \frac{t_k}{h} \right\rfloor$. If the number of discrete intervals $N_m$ is chosen such that
\begin{equation}
    \begin{aligned}
        \label{equ:unitary:6}
        N_m\ge \frac{Me^{\tau_2}(A_{1,\max}A_{2,\max}t+A_{2,\max}(\tau_2+1))}{\varepsilon},
    \end{aligned}
\end{equation}
then the sum $\sum\limits_{k=0}^M(-i)^kX_k$ can be approximated by $\sum\limits_{k=0}^M(-i)^kY_k$ within an error bound of $\varepsilon$:
\begin{equation*}
    \begin{aligned}
        \left\|\sum\limits_{k=0}^M(-i)^kY_k-\sum\limits_{k=0}^M(-i)^kX_k\right\|\le\varepsilon,
    \end{aligned}
\end{equation*}
where $Y_k$, which does not contain the time-ordering operator $\mathcal{T}$, is defined as
\begin{equation*}
    \begin{aligned}
        Y_k=e^{-A_1t}h^k\sum\limits_{0\le m_1<\cdots<m_k< N_m}A_a(m_kh)\cdots A_a(m_1h),
    \end{aligned}
\end{equation*}
in which $m_1,\cdots,m_k$ are indices in $0,\cdots,N_m$.
\begin{proof}
\par According to the proof in Lemma~\ref{lemma:B3}, the approximation of $X_k$ consists of two parts: the first part handles the time dependence of $A_a(t)$, and the second part deals with the integration $\int_0^{t_{j+1}}$.\\
    $\bullet$ {\bf Step 1: the time dependence of $A_a(t)$.} In the proof of Step 1 of Lemma~\ref{lemma:B3}, we need to introduce the following intermediate variable $Z_k$:
\begin{equation}
    \begin{aligned}
        \label{equ:unitary:7}
        Z_k&=e^{-A_1t}\int_0^t\d t_kA_a(m_kh)
        \cdots\int_0^{t_2}\d t_1A_a(m_1h)\\
        &=\int_0^t\cdots\int_0^{t_2}e^{-A_1(t-m_kh)}A_2e^{-A_1(m_k-m_{k-1})h}\\
        &\qquad\times\cdots e^{-A_1(m_2-m_1)h}A_2e^{-A_1m_1h}\d t_1\cdots\d t_k,
    \end{aligned}
\end{equation}
and the key point lies in the magnitude of $\|X_k-Z_k\|$, with an upper bound
\begin{widetext}
    \begin{align*}
        \|X_k-Z_k\|&\le \sum\limits_{k=1}^M\left(\int_0^{t_{k+1}}\cdots\int_0^{t_2}\|e^{-A_1t}A_a(t_p)\cdots A_a(t_2)e^{A_1t_2}\|\cdot\|e^{-A_1t_2}(A_a(t_1)-A_a(m_1h))\|\d t_1\cdots \d t_k\right.\\
        &\qquad\qquad+\sum\limits_{j=2}^{k-1}\int_0^{t_{k+1}}\cdots\int_0^{t_2}\|e^{-A_1t}A_a(t_p)\cdots A_a(t_{j+1})e^{A_1t_j}\|\\
        &\qquad\qquad\qquad\cdot\|e^{-A_1t_{j+1}}(A_a(t_j)-A_a(m_jh))e^{A_1m_{j-1}h}\|\cdot\|e^{-A_1m_{j-1}h}A_a(m_{j-1}h)\cdots A_a(m_1h)\|\d t_1\cdots \d t_k\\
        &\qquad\qquad+\left.\int_0^{t_{k+1}}\cdots\int_0^{t_2}\|e^{-A_1t}(A_a(t_k)-A_a(m_kh))e^{A_1m_{k-1}h}\|\cdot\|e^{-A_1m_{k-1}h}A_a(m_{k-1}h)\cdots A_a(m_1h)\|\d t_1\cdots \d t_k\right)\\
        &\le hA_{1,\max}A_{2,\max}^k\sum\limits_{j=1}^k\left(\int_0^t\d t_k\cdots\int_0^{t_{j+1}}\d t_j\cdots\int_0^{t_2}\d t_1\right)\\
        &\le hA_{1,\max}\frac{\tau_2^k}{(k-1)!},
    \end{align*}
\end{widetext}
where we use the fact $\frac{\d}{\d x}e^{A_1x}A_2e^{-A_1x}=e^{A_1x}[A_1,A_2]e^{-A_1x}$. Therefore, using this upper bound estimate, we can obtain the error between the sums of $X_k$ and $Z_k$ as follows:
\begin{equation}
    \begin{aligned}
        \label{equ:unitary:8}
        &\left\|\sum\limits_{k=0}^M(-i)^kX_k-\sum\limits_{k=0}^M(-i)^kZ_k\right\|\le\sum\limits_{k=1}^M\|X_k-Z_k\|\\
        \le& hA_{1,\max}\sum\limits_{k=1}^M\frac{\tau_2^k}{(k-1)!}<hte^{\tau_2}A_{1,\max}A_{2,\max}.
    \end{aligned}
\end{equation}
$\bullet$ {\bf Step 2: the integration $\int_0^{t_{j+1}}$.} For the treatment of the upper limit of integration, we also adopt an approach similar to Lemma~\ref{lemma:B3}. First, we rewrite $Y_k$ in the following form:
\begin{equation}
    \begin{aligned}
        \label{equ:unitary:9}
        Y_k&=e^{-A_1t}\int_0^t\d t_kA_a(m_kh)
        \cdots\int_0^{m_2h}\d t_1A_a(m_1h)\\
        &=\int_0^t\cdots\int_0^{m_2h}e^{-A_1(t-m_kh)}A_2e^{-A_1(m_k-m_{k-1})h}\\
        &\qquad\times\cdots e^{-A_1(m_2-m_1)h}A_2e^{-A_1m_1h}\d t_1\cdots\d t_k,
    \end{aligned}
\end{equation}
and the key to this proof also lies in estimating the upper bound of $\|Z_k-Y_k\|$:
\begin{widetext}
    \begin{align*}
        \|Z_k-Y_k\|&\le \sum\limits_{k=1}^M\left(\int_0^{t_{k+1}}\cdots\int_0^{t_3}\|e^{-A_1t}A_a(m_kh)\cdots A_a(m_2h)e^{A_1m_2h}\|\d t_2\cdots \d t_k\cdot\int_{m_2h}^{t_2}\|e^{-A_1m_2h}A_a(m_1h)\|\d t_1\right.\\
        &\quad\qquad+\sum\limits_{j=2}^{k-1}\int_0^{t_{k+1}}\cdots\int_0^{t_{j+2}}\int_0^{t_{j+1}}\cdots\int_0^{t_2}\|e^{-A_1t}A_a(m_kh)\cdots A_a(m_{j+1}h)e^{A_1m_{j+1}h}\|\\
        &\quad\qquad\qquad\cdot\|e^{-A_1m_{j-1}h}A_a(m_{j-1}h)\cdots A_a(m_1h)\|\d t_1\cdots \d t_{j+1}\d t_{j-1} \d t_k\int_{m_{j+1}h}^{t_{j+1}}\|e^{-A_1m_{j+1}h}A_a(m_jh)e^{A_1m_{j-1}h}\|\d t_j\\
        &\quad\qquad+\left.\int_{m_{k+1}h}^{t_{k+1}}\|e^{-A_1t}A_a(m_kh)e^{A_1m_{k-1}h}\|\d t_k\cdot\int_0^{t_{k+1}}\cdots\int_0^{t_2}\cdot\|e^{-A_1m_{k-1}h}A_a(m_{k-1}h)\cdots A_a(m_1h)\|\d t_1\cdots \d t_{k-1}\right)\\
        \le& hA_{2,\max}^k\frac{kt^{k-1}}{(k-1)!}.
    \end{align*}
\end{widetext}
By setting $t_{k+1}=t$ and $m_{k+1}=\lfloor\frac{t}{h}\rfloor$, the difference between the sums of $Z_k$ and $Y_k$ is calculated as follows:

\begin{equation}
    \begin{aligned}
        \label{equ:unitary:10}
        &\left\|\sum\limits_{k=0}^M(-1)^kZ_k-\sum\limits_{k=0}^M(-1)^kY_k\right\|\\
        \le&\sum\limits_{k=1}^M\left\|Z_k-Y_k\right\|
        \le h\sum\limits_{k=1}^MA_{2,\max}^k\frac{kt^{k-1}}{(k-1)!}\\
        \le& he^{\tau_2}A_{2,\max}(\tau_2+1).
    \end{aligned}
\end{equation}
\par Combining the results from Eqs. (\ref{equ:unitary:8}) and (\ref{equ:unitary:10}), we can obtain an upper bound for the difference between the sums of $X_k$ and $Y_k$:
\begin{equation*}
    \begin{aligned}
        &\left\|\sum\limits_{k=0}^M(-i)^kX_k-\sum\limits_{k=0}^M(-i)^kY_k\right\|\\
        &\qquad\le he^{\tau_2}(A_{1,\max}A_{2,\max}t+A_{2,\max}(\tau_2+1)).
    \end{aligned}
\end{equation*}
To ensure this upper bound is controlled by the error tolerance $\varepsilon$, $N_m$ must satisfy the following condition:
\begin{equation*}
    \begin{aligned}
        N_m\ge \frac{Me^{\tau_2}(A_{1,\max}A_{2,\max}t+A_{2,\max}(\tau_2+1))}{\varepsilon}.
    \end{aligned}
\end{equation*}
This completes the proof.
\end{proof}
\end{lemma}

\begin{remark}
\label{lemma:C4}
\par For the approximation of $e^{-At}$:
\begin{equation}
    \begin{aligned}
        \label{equ:unitary:11}
        e^{-At}&\approx\left[\sum\limits_{k=0}^M(-i)^k\left(\frac{h}{N_t}\right)^k\right.\\
        &\left.\sum\limits_{0\le m_1<\cdots<m_k< N_m}\left(\overleftarrow{\prod_{j=1}^k}e^{-\frac{A_1(m_{j+1}-m_j)h}{N_t}}A_2\right)\cdot e^{-\frac{A_1m_1h}{N_t}}\right]^{N_t},
    \end{aligned}
\end{equation}
where $m_{k+1}h=t$, we still need to concretely implement $e^{-A_1 \frac{m_jh}{N_t}}$. To simplify notation, set $r_j = \frac{m_jh}{N_t}$.  
Using the fast-forwarding of Hermitian operators from Section~\ref{section:4-1}, we have
\begin{equation}
    \begin{aligned}
        \label{equ:unitary:12}
        e^{-A_1 r_j}
        &= \frac{1}{2\sqrt{\pi}}\int_{-\infty}^{+\infty} e^{-\frac{\eta^2}{4}} e^{i\eta \sqrt{A_1 r_j}}d\eta \\
        &= \frac{1}{2\sqrt{\pi}}\int_{-R}^{R} e^{-\frac{\eta^2}{4}} e^{i\eta \sqrt{A_1 r_j}}d\eta + \mathcal{O}(\varepsilon) \\
        &= \frac{1}{2\sqrt{\pi}}h_a\sum\limits_{k=0}^{N_A-1} e^{-\frac{\eta_k^2}{4}} e^{i\eta_k\sqrt{A_1 r_j}} + \mathcal{O}(\varepsilon),
    \end{aligned}
\end{equation}
where $R$ is the integration cut off, $N_A = \frac{2R}{h_a}$, $\eta_k=-R+h_a k$ and $h_a$ is the step size. Because $e^{i\eta \sqrt{A_1 r_j}}$ is unitary with norm always equal to $1$, the integrand is bounded in norm by the Gaussian factor alone; hence the choice of the cut off $R$ and the Riemann sum discretisation can be made independent of the parameter $r_j$, a very convenient property.

\par By examining the structure of Eq.~(\ref{equ:unitary:11}), each product chain requires error control on at most $N_tM$ occurrences of $e^{-A_1 r_j}$. To guarantee an overall error smaller than $\varepsilon$, each individual $e^{-A_1 r_j}$ must be approximated with error at most $\frac{\varepsilon}{N_tM}$. Using the results of Section~\ref{section:4-1}, this translates to a cut off at most $\mathcal{O}\bigl(\sqrt{A_1 r_j \log\frac{\varepsilon}{N_tM}}\bigr) = \mathcal{O}\bigl(\sqrt{A_1 t \log\varepsilon}\bigr)$. Mimicking the proof of Lemma~\ref{lemma:C3} for the Riemann sum error, one may choose $N_A = \mathcal{O}\bigl(\frac{1}{\varepsilon}\bigr)$; we omit the details here.

\par We can now build a super oracle that provides all the required block-encodings of $e^{-A_1 r}$ in a single unified circuit.  
Define the PREPARE operator (acting on an ancilla register $a$) as
\begin{equation}
    \begin{aligned}
        \label{equ:unitary:13}
        P\ket{0}_a=\frac{1}{\sqrt{\lambda}} \sum_{k=0}^{N_A-1} e^{-\frac{\eta_k^2}{8}} \ket{k}_a,
        \quad 
        \lambda = \sum_{k=0}^{N_A-1} e^{-\frac{\eta_k^2}{4}},
    \end{aligned}
\end{equation}
and let the controlled SELECT operator be
\begin{equation}
    \begin{aligned}
        \label{equ:unitary:14}
        \mathrm{SEL} = \sum_{j=0}^{N_m-1} \ket{j}\bra{j}_c \otimes 
        \left(\sum_{k=0}^{N_A-1} \ket{k}\bra{k}_a \otimes 
        e^{i\eta_k \sqrt{A_1 r_j}} \right),
    \end{aligned}
\end{equation}
where the control register $c$ encodes the desired time parameter $r_j$.  The overall super oracle
\begin{equation}
    \begin{aligned}
        \label{equ:unitary:15}
        \textrm{HAM}_A = (I_c \otimes P^\dagger \otimes I_s) \mathrm{SEL} (I_c \otimes P \otimes I_s)
    \end{aligned}
\end{equation}
then satisfies, after post selecting the ancilla on $\ket{0}_a$,
\begin{equation}
    \begin{aligned}
        \label{equ:unitary:16}
        &(\bra{0}_a \otimes I_c \otimes I_s)\textrm{HAM}_A(\ket{0}_a \otimes I_c \otimes I_s)\\
        =& \sum\limits_{j=0}^{N_m-1}|j\rangle\langle j|_c\otimes\frac{1}{\lambda} \sum_{k=0}^{N_A-1} e^{-\frac{\eta_k^2}{4}} e^{i\eta_k \sqrt{A_1 r_j}}\\
        \approx& \sum\limits_{j=0}^{N_m-1}|j\rangle\langle j|_c\otimes e^{-A_1 r_j}.
    \end{aligned}
\end{equation}
Thus a single call to $\textrm{HAM}_A$, with the control register set to the appropriate $\ket{j}$, yields a block encoding of $e^{-A_1 r_j}$ for any of the required time parameters appearing in Eq.~(\ref{equ:unitary:11}). All further combinations that build the bracket operator $[\cdots]^{N_t}$ can be expressed solely in terms of queries to this super oracle (together with a block encoding of $A_2$), making the query complexity directly the number of invocations of $\textrm{HAM}_A$.
\end{remark}

\subsection{Beyond Square-Root Fast-Forwarding}
\label{section:10-3}
\par In Eq.~(\ref{equ:fast:2}), we presented a square-root fast-forwarding algorithm intended for time-independent purely dissipative (Hermitian) operators. Its core logic centers on the Fourier transform of $e^{-x^2}$ given in Eq.~(\ref{equ:fast:1}). Looking back to Eq.~(\ref{equ:LCHS:1}), LCHS also hinges fundamentally on the Fourier transform of $e^{-x} (x \ge 0)$. This prompts a natural question: Can we generalize the Fourier transform of $e^{-x^\alpha} (x \ge 0)$, modeled as
\begin{equation}
    \begin{aligned}
        e^{-x^\alpha} = \int_{-\infty}^{+\infty} \gamma(\eta) e^{-i \eta x} \d\eta,
    \end{aligned}
\end{equation}
where $\alpha \ge 1$, to construct fast-forwarding protocols reaching beyond the square-root acceleration limit? 
\begin{figure}[tbp]
    \centering
    \includegraphics[width=\linewidth]{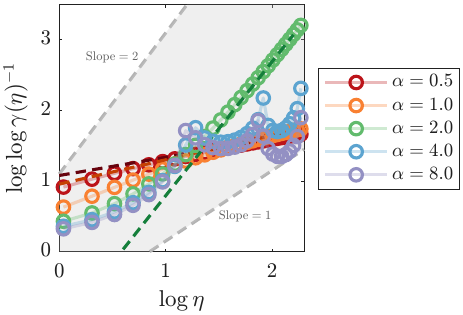}
    \caption{
    \label{figure:fourier-coefficient-decay}\textbf{Decay of Fourier coefficients for $e^{-|x|^\alpha}$.} The curves compare $\log\log\gamma(\eta)^{-1}$ against $\log\eta$ for $\alpha=0.5,1,2,4,8$. The Gaussian case $\alpha=2$ is the steepest stable one, while $\alpha>2$ becomes oscillatory and does not improve the asymptotic decay.}
\end{figure}

\begin{figure*}[htbp]
    \centering
    \includegraphics[width=0.75\linewidth]{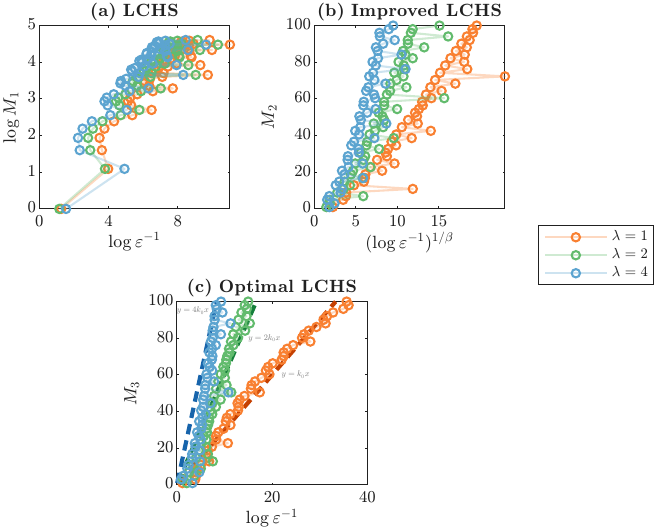}
    \caption{
    \label{figure:lchs-scaling}\textbf{Tolerance dependence of LCHS variants.} Panels (a)--(c) show the original, improved, and optimal LCHS scalings for $\lambda=1,2,4$. Even after optimization, the truncation cost still depends multiplicatively on the target tolerance.}
\end{figure*}

\par To verify this, we perform an inverse Fourier transform on $e^{-|x|^\alpha}$ to examine the properties of $\gamma(\eta)$:
\begin{equation}
    \begin{aligned}
        \gamma(\eta) = \int_{-\infty}^{+\infty} e^{-|x|^\alpha} e^{-i \eta x} \d x.
    \end{aligned}
\end{equation}
When $\alpha = 1$, $\gamma$ represents the probability density function of the Cauchy distribution~\cite{An2023QuantumAF}. When $\alpha = 2$, it gives the Gaussian distribution~\cite{Jin2026Transmutation,Kharazi2026Sublinear}. For $1 < \alpha < 2$, it forms the L\'evy distribution~\cite{An2023Optimal}. However, when $\alpha > 2$, $\gamma$ is no longer real-valued and nonnegative, which prevents its use as a standard probability density function on the real line. In Fig.~\ref{figure:fourier-coefficient-decay}, we plot the real part of $\log \log \gamma^{-1}$ against $\log \eta$. A steeper slope indicates a faster decay rate. We observe that the decay rate uniquely maximizes at $\alpha = 2$. Furthermore, the curves for $\alpha > 2$ exhibit severe numerical oscillations. This experiment demonstrates that simply generalizing the Fourier transform for $e^{-x^\alpha} (x \ge 0)$ cannot surpass existing limits.

\section{Further Details on Existing Methods}
\label{section:11}

\subsection{Why LCHS Is Not Strictly Additive}
\label{section:11-1}
\par We now give the Paley-Wiener argument behind the discussion in Section~\ref{section:6-1}. The key point is that the Fourier cutoffs in LCHS must grow with the target tolerance, so the method cannot achieve a strictly additive non-unitary overhead.
\begin{lemma}
\label{lemma:D1}
\par For $\sigma > 0$, suppose a function $f(t)\in L^2(-\infty,+\infty)$ has support in $[-\sigma, \sigma]$. Its Fourier transform $F(x)$ can be expanded into an entire function $F(z)$. Moreover, there exist constants $C,N,R$ such that
\begin{equation}
    \begin{aligned}
        \label{equ:LCHS:5}
        |F(z)|\le C(1+|z|)^Ne^{R|\mathrm{Im}z|}.
    \end{aligned}
\end{equation}
\qed
\end{lemma}
\noindent Consider the scalar case $L=\lambda I$ with $\lambda>0$, where LCHS assumes
\begin{equation}
    \begin{aligned}
        \label{equ:LCHS:6}
        e^{-\lambda t}=\int_{-\infty}^{+\infty}\gamma(\eta)e^{-i\eta \lambda t}\d\eta.
    \end{aligned}
\end{equation}
To achieve error tolerance $\varepsilon$, one truncates the integral at $-M_1$ and $M_2$ so that
\begin{equation}
    \begin{aligned}
        \label{equ:LCHS:7}
        \int_{-\infty}^{-M_1}\gamma(\eta)e^{-i\eta\lambda t}\d\eta+\int_{M_2}^{+\infty} \gamma(\eta) e^{-i\eta\lambda t}\d\eta<\varepsilon.
    \end{aligned}
\end{equation}
If $M_1$ and $M_2$ were bounded independently of $\varepsilon$, then $\gamma(\eta)$ would have compact support. Lemma~\ref{lemma:D1} would then imply that its Fourier transform extends to an entire function with growth bounded by Eq.~(\ref{equ:LCHS:5}). But along the path $z=-t$ with $t>0$, the left-hand side becomes $e^{\lambda t}$, whose exponential growth is incompatible with that bound. Therefore the truncation boundaries must depend on $\varepsilon$, and implementing the extremal phases $e^{iM_1\lambda t}$ and $e^{-iM_2\lambda t}$ inevitably induces a multiplicative tolerance-dependent overhead. This is the precise sense in which LCHS is near-optimal but not strictly additive. Fig.~\ref{figure:lchs-scaling} illustrates the same obstruction numerically.

\par Fig.~\ref{figure:lchs-scaling} illustrates the same obstruction numerically across representative LCHS variants.

\subsection{Full Block-Matrix Embedding for NDME}
\label{section:11-2}
\par Here we give the full block-matrix realization omitted from Section~\ref{section:6-2}. Based on phase-driven APS and the interaction picture for Lindbladians, define
\begin{equation}
    \begin{aligned}
        u(t)=\mathcal{U}_p^\dagger(t)u_p(t),\quad \tilde{\rho}(t)=\mathcal{U}_p^\dagger(t)\tilde{\rho}_p(t)\mathcal{U}_p(t),
    \end{aligned}
\end{equation}
where $u_p(t)$ and $\tilde{\rho}_p(t)$ are the corresponding interaction-picture variables and $\mathcal{U}_p(t)$ is defined as in Section~\ref{section:2}. The corresponding Lindblad model for $\tilde{\rho}_p(t)$ uses a single jump operator $\tilde{F}_p(t)=\mathcal{U}_p^\dagger(t)\tilde{F}(t)\mathcal{U}_p(t)$. This suggests combining $\tilde{\rho}(t)$ and $u(t)$ into the block matrix
\begin{equation}
    \begin{aligned}
        \rho(t)
        &=\begin{bmatrix}
            \tilde{\rho}(t) & u(t)\\
            u^\dagger(t) & 1
        \end{bmatrix}
        =\begin{bmatrix}
            \mathcal{U}_p^\dagger(t) & O\\
            O & 1
        \end{bmatrix}\rho_p(t)\begin{bmatrix}
            \mathcal{U}_p(t) & O\\
            O & 1
        \end{bmatrix},
    \end{aligned}
\end{equation}
where
\begin{equation}
    \begin{aligned}
        \rho_p(t)=\begin{bmatrix}
            \tilde{\rho}_p(t) & u_p(t)\\
            u_p^\dagger(t) & 1
        \end{bmatrix}.
    \end{aligned}
\end{equation}
This is the interaction picture generated by $H(t)=\begin{bmatrix}A_2(t) \\ & 0\end{bmatrix}$ together with the constraint $\mathrm{tr}[\tilde{\rho}(t)]=0$. Its intrinsic evolution equation is
\begin{equation*}
    \begin{aligned}
        &\frac{\d\rho_p(t)}{\d t}
        =\begin{bmatrix}
            \tilde{F}_p(t) & O\\
            O & 0
        \end{bmatrix}\rho_p(t)
        \begin{bmatrix}
            \tilde{F}_p^\dagger(t) & O\\
            O & 0
        \end{bmatrix}
        +\rho_p(t)\begin{bmatrix}
            -A_p(t) & O\\
            O & 0
        \end{bmatrix}\\
        +&\begin{bmatrix}
            -A_p(t) & O\\
            O & 0
        \end{bmatrix}\rho_p(t)
        +\begin{bmatrix}
            \left\{\rho_p(t),-\frac{1}{2}\tilde{F}_p^\dagger(t) \tilde{F}_p(t)+A_p(t)\right\} & O\\
            O & 0
        \end{bmatrix},
    \end{aligned}
\end{equation*}
where $A_p(t)=\frac{1}{2}\tilde{F}_p^\dagger(t)\tilde{F}_p(t)$. Defining
\begin{equation}
    \begin{aligned}
        F_p(t)=\begin{bmatrix}\tilde{F}_p(t) & O \\ O & 0 \end{bmatrix},\quad F(t)=\begin{bmatrix}\sqrt{2A_1(t)} & O \\ O & 0 \end{bmatrix},
    \end{aligned}
\end{equation}
this becomes the compact Lindbladian equation quoted in Eq.~(\ref{equ:lindbladian:6}):
\begin{equation*}
    \begin{aligned}
        \frac{\d\rho_p(t)}{\d t}
        &=F_p(t)\rho_p(t)F_p^\dagger(t)-\frac{1}{2}\{\rho_p(t),F_p^\dagger(t)F_p(t)\}.
    \end{aligned}
\end{equation*}
Finally, applying the inverse transformation from Eq.~(\ref{equ:lindbladian:3}) embeds Eq.~(\ref{equ:1}) into the canonical NDME master equation with
\begin{equation}
    \begin{aligned}
        H(t)=\begin{bmatrix} A_2(t) \\ & 0 \end{bmatrix},\quad F(t)=\begin{bmatrix}\sqrt{2A_1(t)} & O \\ O & 0 \end{bmatrix}.
    \end{aligned}
\end{equation}
Under APS, NDME is therefore not an independent conceptual mechanism; it is the block-matrix form taken by the dissipative part after the coherent contribution has been separated out.
\end{document}